\documentclass{LMCS}
\usepackage{logiclmcs}
\usepackage{graphicx}
\usepackage{xspace}
\usepackage{url}
\usepackage{soul,color}
\usepackage[utf8]{inputenc}
\usepackage{enumerate,hyperref}

\def\doi{6 (4:1) 2010}
\lmcsheading%
{\doi}
{1--26}
{}
{}
{Sep.~17, 2010}
{Oct.~20, 2010}
{}   

\begin{document}



\newcommand{\hhl}[1]{\hl{#1}}

\newcommand{\debug}[1]{#1}

\addtolength\marginparwidth{2cm}
\newcommand{\luc}[1]{\debug{\marginpar{\tiny {\bf Luc}: #1}}}
\newcommand{\mikolaj}[1]{\debug{\marginpar{\tiny Mikolaj: #1}}}

\newcommand{\botcomp}{{H_\bot}}

\newcommand{\reorder}{\mathrm{reorder}}
\newcommand{\lgreensim}{\sim_{\mathcal L}}
\newcommand{\rgreensim}{\sim_{\mathcal R}}
\newcommand{\lex}{<_\mathrm{lex}}
\newcommand{\lexeq}{\leq_\mathrm{lex}}
\newcommand{\bigpiece}{\mathcal T}
\newcommand{\miff}{\mbox{iff}}
\newcommand{\pen}{n}
\newcommand{\hole}{\Box}
\newcommand{\ltrivial}{\mathcal L}
\newcommand{\rtrivial}{\mathcal R}

\newcommand{\rlogic}{\mathcal R}
\newcommand{\llogic}{\mathcal L}

\newcommand{\tlf}{\mathsf{F}}
\newcommand{\tle}{\mathsf E}
\newcommand{\tlu}{\mathsf U}
\newcommand{\tlex}{\mathsf{EX}}
\newcommand{\tlspider}[1]{\mathsf{E}^{#1}\mathsf{X}}
\newcommand{\tlef}{\mathsf{EF}p}
\newcommand{\tlleaf}{\mathsf{L}^{\downarrow}}
\newcommand{\tlfup}{\mathsf{F}^{\uparrow}}
\newcommand{\tlfh}{\mathsf{F}^{\rightarrow}}
\newcommand{\tlfhb}{\mathsf{F}^{\leftarrow}}
\newcommand{\tlag}{\mathsf{AG}}
\newcommand{\freef}[1]{#1^\Delta}
\newcommand{\tlsib}{\mathsf S}

\newcommand{\Stwo}{\ensuremath{\Sigma_2(<)}\xspace}
\newcommand{\Ptwo}{\ensuremath{\Pi_2(<)}\xspace}
\newcommand{\Dtwo}{\ensuremath{\Delta_2(<)}\xspace}
\newcommand{\Stwol}{\ensuremath{\Sigma_2(<,\lex)}\xspace}
\newcommand{\Ptwol}{\ensuremath{\Pi_2(<,\lex)}\xspace}
\newcommand{\Dtwol}{\ensuremath{\Delta_2(<,\lex)}\xspace}

\newcommand{\lgreen}{\le_\ltrivial}
\newcommand{\rgreen}{\le_\rtrivial}
\newcommand{\lgreeng}{\ge_\ltrivial}
\newcommand{\rgreeng}{\ge_\ltrivial}

\newcommand{\rsim}{\sim_{\rtrivial}}

\newcommand{\pieces}{\mathit{pieces}}
\newcommand{\subforest}{\leq}
\newcommand{\piece}{\preceq}
\newcommand{\pieceneq}{\prec}

\newcommand{\stab}{\mathit{stab}}
\newcommand{\sstab}{\mathit{sstab}}
\newcommand{\inter}{\Delta_2}

\newcommand{\mand}{\ \mbox{and}\ } \newcommand{\first}{\mathit{new}}

\title{Tree languages defined in first-order logic with one quantifier
alternation}
\author[M.~ Boja{\'n}czyk]{Miko{\l}aj Boja{\'n}czyk\rsuper a}	
\address{{\lsuper a}Warsaw University}	
\email{bojan@mimuw.edu.pl}  
\thanks{{\lsuper a}Supported by Polish government grant
 no. N206 008 32/0810.}	

\author[L.~Segoufin]{Luc Segoufin\rsuper b}	
\address{{\lsuper b}INRIA - LSV}
\urladdr{http://www-rocq.inria.fr/$\tilde{\ }$segoufin}
\thanks{{\lsuper b}Work partially funded by the AutoMathA
programme of the ESF}	

\begin{abstract}
  We study tree languages that can be defined in $\Delta_2$. These are tree
  languages definable by a first-order formula whose quantifier prefix is
  $\exists^*\forall^*$, and simultaneously by a first-order formula whose
  quantifier prefix is $\forall^*\exists^*$. For the quantifier free part we
  consider two signatures, either the descendant relation alone or together
  with the lexicographical order relation on nodes. We provide an effective
  characterization of tree and forest languages
  definable in $\Delta_2$. This characterization is in terms of algebraic
  equations. Over words, the class of word languages definable in $\Delta_2$
  forms a robust class, which was given an effective algebraic characterization
  by Pin and Weil~\cite{weilpinpoly}.
\end{abstract}

\maketitle

\section{Introduction}

We say a logic $\Ll_1$ has a decidable characterization inside a logic $\Ll_2$ if the following
decision problem is decidable: ``given as input a formula of the logic $\Ll_2$,
decide if it is equivalent to some formula of the logic $\Ll_1$''. We are interested in the case when the logic $\Ll_2$ is MSO on words or trees, and $\Ll_1$ represents some fragment of $\Ll_2$. 

This type of problem has been successfully studied in the case when $\Ll_2$ is
MSO on finite words. In other words $\Ll_2$, represents the class of regular
word languages.  Arguably best known is the result of McNaughton, Papert and
Sch\"utzenberger~\cite{schutzenberger,mcnaughton}, which says that the
following two conditions on a regular word language $L$ are equivalent: a) $L$
can be defined in first-order logic with order and label tests; b) the
syntactic semigroup of $L$ does not contain a non-trivial group. Since
condition b) can be effectively tested, the above theorem gives a decidable
characterization of first-order logic. This result demonstrates the importance
of this type of work: a decidable characterization not only gives a better
understanding of the logic in question, but it often reveals unexpected
connections with algebraic concepts.  During several decades of research,
decidable characterizations have been found for fragments of first-order logic
with restricted quantification and various signatures (typically subsets of the
order relation and the successor relation), as well as a large group of
temporal logics, see \cite{pin-survey} and \cite{wilke} for references.

An important part of this research has been devoted to the quantifier
alternation hierarchy, where each level counts the alterations between
$\forall$ and $\exists$ quantifiers in a first-order formula in prenex normal
form. The quantifier free part of such a formula is  built using a binary
predicate $<$ representing the linear order on the word. Formulas that have
$n-1$ alternations (and therefore $n$ blocks of quantifiers) are called $\Sigma_n(<)$ if they begin with $\exists$, and
$\Pi_n(<)$ if they begin with~$\forall$. For instance, the word property ``some
position has label $a$'' can be defined by a $\Sigma_1(<)$ formula $\exists x.\
a(x)$, while the language ``nonempty words with at most two positions that
do not have label $a$'' can be defined by the $\Sigma_2(<)$ formula
\[
  \exists x_1 \exists x_2  \forall y \quad (y \neq x_1 \land y \neq
  x_2) \quad 
  \Rightarrow \quad a(y)\ .
\]

A lot of attention has been devoted to analyzing the low levels of the
quantifier alternation hierarchy for word languages.  The two lowest levels are easy: a word
language is definable in $\Sigma_1(<)$ (resp.~$\Pi_1(<)$) if and only if it is
closed under inserting (removing) letters. Both properties can be tested in
polynomial time based on a recognizing automaton, or semigroup.  However, just
above $\Sigma_1(<),\Pi_1(<)$, and even before we get to $\Sigma_2(<),
\Pi_2(<)$, we already find two important classes of languages.  A fundamental
result, due to Simon~\cite{simonpiecewise}, says that a language is defined by
a Boolean combination of $\Sigma_1(<)$ formulas if and only if its syntactic
monoid is $\mathcal J$-trivial. Above the Boolean combination of $\Sigma_1(<)$,
we find $\Delta_2(<)$, i.e.~languages that can be defined simultaneously in
$\Sigma_2(<)$ and $\Pi_2(<)$. As we will describe later on, this class turns
out to be surprisingly robust, and it is the focus of this paper. Another
fundamental result, due to Pin and Weil~\cite{weilpinpoly}, says that a regular
language is in $\Delta_2(<)$ if and only if its syntactic monoid is in {\sc
  DA}.  The limit of our knowledge is level $\Sigma_2(<)$: it is decidable if a
language can be defined on level $\Sigma_2(<)$~\cite{arfi91,weilpinpoly}, but
there are no known decidable characterization for Boolean combinations of
$\Sigma_2(<)$, for $\Delta_3(<)$, for $\Sigma_3(<)$, and upwards.

For trees even less is known.  No decidable characterization has been
found for what is arguably the most important proper subclass of
regular tree languages, first-order logic with the descendant
relation, despite several attempts.  Similarly open are chain logic
and the temporal logics CTL, CTL* and PDL.  However, there has been
some recent progress.  In~\cite{efextcs}, decidable characterizations
were presented for some temporal logics, while Benedikt and
Segoufin~\cite{segoufinfo} characterized tree languages definable in
first-order logic with the successor relation (but without the
descendant relation).

This paper is part of a program to understand the expressive power of
first-order logic on trees, and the quantifier alternation hierarchy
in particular. The idea is to try to understand the low levels of the
quantifier alternation hierarchy before taking on full first-order
logic (which is contrary to the order in which word languages were
analyzed). We focus on two signatures. The first signature contains
unary predicates for label tests and the ancestor order on nodes,
denoted $<$. The second signature assumes that the trees have an order
on siblings, which induces a lexicographical linear order on nodes,
denoted $\lex$.  Both signatures generalize the linear order on words.
As shown in~\cite{forestexp}, there is a reasonable notion of
concatenation hierarchy for tree languages that corresponds to the
quantifier alternation hierarchy. Levels $\Sigma_1(<)$ and $\Pi_1(<)$
are as simple for trees as they are for words. A recent
result~\cite{simontrees} extends Simon's theorem to trees, by giving a
decidable characterization of tree languages definable by a Boolean
combination of $\Sigma_1(<)$ formulas, and also a decidable
characterization of Boolean combinations of $\Sigma_1(<,\lex)$
formulas.  There is no known decidable characterization of tree
languages definable in $\Sigma_n(<)$ for $n \ge 2$.

The contribution of this paper is a decidable characterization of tree
languages definable in \Dtwo, i.e.~definable both in \Stwo
and $\Pi_2(<)$.  We also provide a decidable characterization of tree
languages definable in \Dtwol.

As we signaled above, for word languages the class
$\Delta_2(<)$ is well studied and important, with numerous equivalent
characterizations.  Among them one can
find~\cite{weilpinpoly,therienwilkefo2,turtle,EVW02}: a) word
languages that can be defined in the temporal logic with operators
$\tlf$ and $\tlf^{-1}$; b) word languages that can be defined by a
first-order formula with two variables, but with unlimited quantifier
alternations; c) word languages whose syntactic semigroup belongs to
the semigroup variety DA; d) word languages recognized by two-way
ordered deterministic automata; e) a certain form of ``unambiguous''
regular expressions.

It is not clear how to extend some of these concepts to trees. Even when
natural tree counterparts exist, they are not equivalent. For instance, the
temporal logic in a) can be defined for trees---by using operators ``in some
descendant'' and ``in some ancestor''.  This temporal logic was studied
in~\cite{fo2tree}, however it was shown to have an expressive power
incomparable with that of \Dtwo. A characterization of \Dtwo was left
as an open problem, one which is solved here.

We provide an algebraic characterization of tree languages definable in \Dtwo
and in \Dtwol. This characterization is effectively verifiable if the language
is given by a tree automaton. It is easy to see that the word setting can be
treated as a special case of the tree setting. Hence our characterization
builds on the one over words.  However the added complexity of the tree setting
makes both formulating the correct condition and generalizing the proof quite
nontrivial.

\section{Trees forests and languages}
In this section, we present some basic definitions regarding trees. We also present the formalism of forest algebra, which is used in  our characterizations.

\subsection{Trees, forests and contexts}
In this paper we work with finite, unranked, ordered trees and forests
over a finite alphabet $A.$ Formally, these are expressions defined
inductively as follows: If $s$ is a forest and $a\in A,$ then $as$ is
a tree.  If $t_1,\ldots,t_n$ is a finite sequence of trees, then
$t_1+\cdots+t_n$ is a forest.  This applies as well to the empty
sequence of trees, which is called the {\it empty forest,} and denoted
0 (and which provides a place for the induction to start).  Forests
and trees alike will be denoted by the letters $s,t,u,\ldots$ When
necessary, we will  remark which forests are trees, i.e.~contain
only one tree in the sequence.


The notion of node, as well as the descendant and ancestor relations are
defined in the usual way. We write $x < y$ to say that $x$ is a strict ancestor of
$y$ or, equivalently, that $y$ is a strict descendant of $x$. As usual, we
write $x \le y$ when $x=y$ or $x < y$.  The parent of a node $x$ is its
immediate ancestor. Two nodes $x$ and $y$ are \emph{siblings} if they have the
same parent.  We also use the lexicographic order on nodes, written $\lex$.
Recall that $x \lex y$ holds if either $x < y$, or there are nodes $x' \le x$
and $y' \le y$ such that $x'$ is a sibling to the left of $y'$.

If we take a forest and replace one of the leaves by a special symbol
$\hole$, we obtain a \emph{context.}  Contexts will be denoted
using letters $p,q,r$. A forest $s$ can be substituted in place of the
hole of a context $p$, the resulting forest is denoted by $ps$.
 There is a natural composition operation on contexts: the
context $qp$ is formed by replacing the hole of $q$ with $p$.
This operation is associative, and satisfies $(pq)s=p(qs)$ for all
forests $s$ and contexts $p$ and $q$.


We say a forest $s$ is an \emph{immediate piece} of a forest $s'$ if
$s,s'$ can be decomposed as $s=pt$ and $s'=pat$ for some context $p$,
some label $a$, and some forest $t$. The reflexive transitive closure
of the immediate piece relation is called the \emph{piece}
relation. We write $s \piece t$ to say that $s$ is a piece of $t$. In
other words, a piece of $t$ is obtained by removing nodes from $t$.
We extend the notion of piece to contexts. In this case, the hole must
be preserved while removing the nodes. The notions of piece for
forests and contexts are related, of course. For instance, if $p$, $q$
are contexts with $p \piece q$, then $p0 \piece q0$. Also, conversely,
if $s \piece t$, then there are contexts $p \piece q$ with $s=p0$ and
$t=q0$. (For instance, one can take $p=\hole + s$ and $q=\hole + t$.)
The picture below depicts two contexts, the left one being a piece of
the right one, as can be seen by removing the white nodes.
\medskip
\begin{center}
  \includegraphics[scale=0.8]{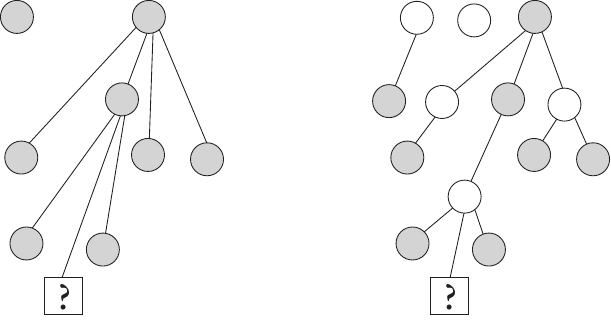}
 \end{center}
We will be considering three types of languages in the paper:
 \emph{forest languages} i.e.~sets of forests, denoted $L$;
 \emph{context languages}, i.e.~sets of contexts, denoted $K$, and
 \emph{tree languages}, i.e.~sets of trees, denoted $M$. Note that a
 forest language can contain trees.

\subsection{Forest algebras}
\label{sec:forest-algebra}
{\it Forest algebras} were introduced by Boja\'nczyk and Walukiewicz as an
algebraic formalism for studying regular tree languages~\cite{forestalgebra}.  Here
we give a brief summary of the definition of these algebras and their important
properties.  A forest algebra consists of a pair $(H,V)$ of finite monoids,
subject to some additional requirements, which we describe below.  We write the
operation in $V$ multiplicatively and the operation in $H$ additively, although
$H$ is not assumed to be commutative.  We accordingly denote the identity of
$V$ by $\hole$ and that of $H$ by 0.
 We require that $V$ act on the left of $H$.  That is, there is a map
 $(h,v)\mapsto vh\in H$
 such that 
 $w(vh)=(wv)h$
 for all $h\in H$ and $v,w\in V.$ We further require that this action be {\it
   monoidal,} that is,
$\hole \cdot h=h$
 for all $h\in H,$ and that it be {\it faithful,} that is,
 if $vh=wh$ for all $h\in H,$ then $v=w.$
Finally we require that for every $g\in H,$ $V$ contains elements $(\hole+g)$
 and $(g+\hole)$ defined by
 $(\hole+g)h=h+g, (g+\hole)h=g+h$
 for all $h\in H.$

 A morphism $\alpha:(H_1,V_1)\to (H_2,V_2)$ of forest algebras is actually a
 pair $(\alpha_H,\alpha_V)$ where $\alpha_H$ is a monoid morphism between $H_1$ and
 $H_2$ and $\alpha_V$ is a monoid morphism between $V_1$ and $V_2$, such that
 $\alpha_H(vh)=\alpha_V(v)\alpha_H(h)$ for all $h\in H,$ $v\in V.$ However, we will
 abuse notation slightly and denote both component maps by $\alpha.$

 Let $A$ be a finite alphabet, and let us denote by $H_A$ the set of
 forests over $A,$ and by $V_A$ the set of contexts over $A.$ Each of
 these is a monoid, with the operations being forest concatenation and
 context composition, respectively. The pair $(H_A,V_A)$, with forest
 substitution as action, forms a forest algebra, which we denote
 $A^{\Delta}.$


 We say that a forest algebra $(H,V)$ {\it recognizes} a forest
 language $L\subseteq H_A$ if there is a morphism
 $\alpha:A^{\Delta}\to (H,V)$ and a subset $X$ of $H$ such that
 $L=\alpha^{-1}(X).$ A forest language is regular, i.e.~recognized by any of the many equivalent notions of automata for unranked trees that can be found in the literature, if and only if it is
 recognized by a finite forest algebra~\cite{forestalgebra}.

 Given any finite monoid $M$, there is a number $\omega(M)$ (denoted by
 $\omega$ when $M$ is understood from the context) such that for all elements
 $x$ of $M$, $x^\omega$ is an idempotent: $x^\omega=x^\omega x^\omega$.
 Therefore for any forest algebra $(H,V)$ and any element $u$ of $V$ and $g$ of
 $H$ we will write $u^\omega$ and $\omega g$ for the corresponding
 idempotents. The element $u^\omega$ is idempotent with respect to the
 operation in $V$. The element $\omega g$, which is the same as
 $(\hole+g)^\omega 0$ and the same as $(g+\hole)^\omega 0$, is idempotent with
 respect to the operation in $H$.

 Given a forest language $L\subseteq H_A$ we define an equivalence relation $\sim_L$ on
 $H_A$ by setting $s\sim_L s'$ if and only if for every context $p\in
 V_A,$ the forests $ps$ and $ps'$ are either both in $L$ or both outside of $L.$
 We further define an equivalence relation on $V_A$, also denoted
 $\sim_L,$ by setting $p\sim_L p'$ if for all $s\in H_A,$ $ps\sim_L
 p's.$ This pair of equivalence relations defines a congruence of
 forest algebras on $A^{\Delta},$ and the quotient $(H_L,V_L)$ is
 called the {\it syntactic forest algebra} of $L.$ Each equivalence
 class of $\sim_L$ is called a \emph{type}.

 We now extend the notion of piece to elements of a forest algebra $(H,V)$. The
 general idea is that a context type $v \in V$ is a piece of a context type $w \in V$ if
 one can construct a term (using elements of $H$ and $V$) which evaluates to
 $w$, and then take out some parts of this term to get $v$.
\begin{def}\label{df:piece-alg}
  Let $(H,V)$ be a forest algebra. We say $v \in V$ \emph{is a piece}
  of $w \in V$, denoted by $v \piece w$, if there is an alphabet $A$
  such that $\alpha(p)=v$ and $\alpha(q)=w$ hold for some morphism
  \[
    \alpha : \freef A \to (H,V)
  \]
  and some contexts $p \piece q$ over $A$.  The relation $\piece$ is
  extended to $H$ by setting $g \piece h$ if $g=v0$ and $h=w0$ for
  some context types $v \piece w$.
\end{def}

\section{Logic}
The focus of this paper is the expressive power of first-order logic
on trees.  A forest can be seen as a logical relational structure.
The domain of the structure is the set of nodes. (We allow empty
domains, which happens when an empty forest $0$ is considered.) We
consider two different signatures. Both of them contain a unary
predicate $P_a$ for each symbol $a$ of the alphabet $A$, as well as a
binary predicate $<$ for the ancestor relation. Furthermore, the
second signature also contains a binary predicate $\lex$ for the
lexicographic order on nodes.  A formula without free variables over
these signatures defines a set of forests, these are the forests where
it is true.  We are particularly interested in formulas of low
quantifier complexity.  A $\Sigma_2$ formula is a formula of the form
\[
  \exists x_1\cdots x_n\ \forall y_1 \cdots y_m~~ \gamma\ ,
\]
where $\gamma$ is quantifier free.  Languages defined in $\Sigma_2$
are closed under disjunction and conjunction, but not necessarily
negation. The negation of a $\Sigma_2$ formula is called a $\Pi_2$
formula, equivalently this is a formula whose quantifier prefix is
$\forall^* \exists^*$.  A forest property is called $\Delta_2$ if it
can be expressed both by a $\Sigma_2$ and a $\Pi_2$ formula.
We will use \Stwo and \Stwol to specify which predicates are used in
the signature, similarly for $\Pi_2$ and $\Delta_2$.

With limited quantification, the choice of signature is a delicate
question. For instance, adding a child relation changes the expressive power.

\subsection{The problem}

We want an algorithm deciding whether a given regular forest language is
definable in \Dtwol and another one for deciding whether it is in \Dtwo.

As noted earlier, the corresponding problem for words was solved by
Pin and Weil~\cite{weilpinpoly}: a word language $L$ is definable in
\Dtwo if and only if its syntactic monoid $M(L)$ belongs to the
variety DA, i.e.~it satisfies the identity
\[
(mn)^\omega = (mn)^\omega m (mn)^\omega
\]
\noindent
for all $m,n\in M(L)$. The power $\omega$ means that
the identity holds for sufficiently large powers (in different
settings, $\omega$ is defined in terms of idempotent powers, but the
condition on sufficiently large powers is good enough here).  Since
one can effectively test if a finite monoid satisfies the above
property (it is sufficient to verify the power $|M(L)|$), it is
decidable whether a given regular word language is definable in
\Dtwo. We assume that the language $L$ is given by its syntactic
monoid and syntactic morphism, or by some other representation, such
as a finite automaton, from which these can be effectively computed.

We will show that a similar characterization can be found for forests;
although the identities will be more involved.  For decidability, it
is not important how the input language is represented. In this paper,
we will represent a forest language by a forest algebra that
recognizes it.  Forest algebras are described in the next section.

\subsection{Tree languages.}
We give an algorithm which says when a forest language belongs to a class $\Ll$, which is either \Dtwo or \Dtwol.
What about tree languages? There are two ways of getting a class of tree languages from a class of forest languages $\Ll$. 
\begin{enumerate}[(1)]
\item The class of tree languages that belong to $\Ll$.
\item The class of tree languages of the form $L \cap T_A$, where $A$ is an alphabet,  $T_A$ is the set of all trees over alphabet $A$, and $L \in \Ll$ is a forest language over $A$.
\end{enumerate}
Our algorithm gives a decision procedure under the first definition. The usual understanding of  tree languages definable in \Dtwo or \Dtwol corresponds to the second definition.

Fortunately, the two  definitions are equivalent when $\Ll$ is either \Dtwo or \Dtwol. This is because in both cases,   $\Ll$ is closed under intersection and contains the languages  $T_A$.  

Closure under intersection is immediate. Why does $\Ll$ contain the languages $T_A$? Since \Dtwo is the less powerful logic, it suffices to show how to define $T_A$ using a \Stwo formula, and also using a \Ptwo formula. The \Stwo
formula says there exists a node that is an ancestor of all other
nodes, while the \Ptwo formula says that for every two nodes, there
exists a common ancestor. 

In general, the definitions of tree language classes are not equivalent. Consider as $\Ll$ the class of forest languages defined by purely existential formulas. In particular, if a language $L \in \Ll$ contains a tree $t$, then it also contains the forest $t+t$. This means that under the first definition, the only tree languages we get are the empty tree languages. Under the second definition, we get some more tree languages, such as ``trees with at least two nodes''.

\subsection{Basic properties of \texorpdfstring{$\Pi_1$}{Pi1} and  \texorpdfstring{$\Sigma_2$}{Sigma2}}
Most of the proofs in the paper will work with \Stwo or \Stwol formulas. We
present some simple properties of such formulas in this section.

Apart from defining forest languages, we will also be using formulas to define
languages of contexts. To define a context language we use formulas with a free
variable; such a formula is said to hold in a context if it is true when the
free variable is mapped to the hole of the context.§ For instance, the formula
$\forall y \ y < x \Rightarrow a(y)$ with a free variable $x$ defines the set
of contexts where every ancestor of the hole has label $a$.

We begin by describing the expressive power of  purely universal formulas $\Pi_1$.
\begin{lem}\label{lem:pi1-lemma}
  A forest language is closed under pieces if
  and only if it is definable in $\Pi_1$.  Likewise for
  context languages.
\end{lem}
\begin{proof}
  It is clear that a forest language definable in $\Pi_1$ is closed under
  pieces as the models of a $\Pi_1$ formula are closed under substructures.

  For the converse, let $L$ be a forest language that is closed under pieces,
  and let $ \alpha : \freef A \to (H,V)$ be its syntactic algebra.  Thanks to a
  pumping argument, any forest has a piece with the same type, but at most
  $|H|^{|H|}$ nodes.  Let $T$ be the finite set of forests with at most
  $|H|^{|H|}$ nodes that are outside $L$. Thanks to the pumping argument, a
  forest belongs to $L$ if and only if it has no piece in the set $T$. The
  latter is a property that can be expressed in $\Pi_1(<, \lex)$.
  
  If, additionally, the language $L$ is commutative, then we do not need
  to worry about the lexicographic order when talking about the pieces
  in $T$, only the descendant order is relevant.
\end{proof}

We now turn to the formulas from $\Sigma_2$. We begin  with \Stwol,
since it has the  better closure properties.
\begin{lem}\label{fact:simple-lex}
  Let $K,K'$ be context languages and $L,L'$ be forest languages. If
  these languages are all definable in \Stwol, then so are
\begin{enumerate}[\em(1)]
\item the forest language $KL=\set{qt: q \in K, t \in L}$,
\item the forest language $L+L' =\set{t+t' : t \in L, t' \in L'}$,
\item the context language $KK' = \set{qq' : q \in K, q' \in K'}$,
\item the forest language $L \cap L' =\set{t : t \in L, t \in L'}$,
\item the context language $K \cap K' = \set{q : q \in K, q \in K'}$,
\end{enumerate}
\end{lem}

\proof
  We only do the proof for $KL$, the others are treated similarly. 
When does a forest $t$ belong to $KL$? There must exist two siblings $x_1$ and $x_2$ such that the set, call it $X$, of gray nodes in the picture below
\medskip
\begin{center}
\includegraphics[scale=0.8]{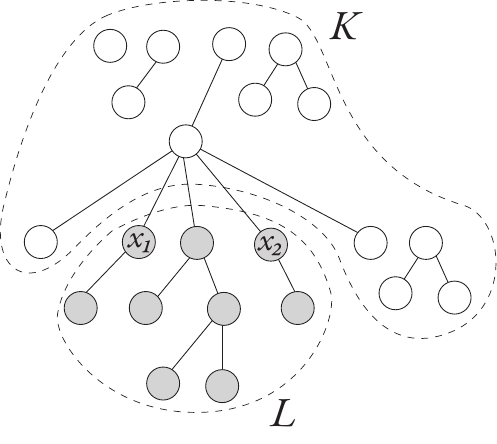}
\end{center}
describes a forest in $L$, and the other nodes describe a context in $K$. Below
we define this property more precisely, and show that it can be defined in
\Stwol.

First, we want to say that the nodes $x_1$ and $x_2$ are siblings. This can be
expressed by a formula, call it $\alpha(x_1,x_2)$, of \Stwol. The formula
quantifies existentially a common ancestor and uses universal quantification to
check that this common ancestor is a parent of both nodes:
\[
\alpha(x_1,x_2) = \exists x~\forall y~ x_1 \lexeq x_2 \wedge  x < x_1 \wedge x < x_2
\wedge  (y < x_1 \rightarrow y \leq x) \wedge (y < x_2 \rightarrow y \leq x_2)
\]

Next, we describe the set $X$.  Membership $x \in X$ is defined by a quantifier-free formula 
\[
	\beta(x) = (x_1 \lexeq x \land x \lexeq x_2) \lor (x > x_2).
\]

Next, we say what it means for the part inside $X$ describes a forest in $L$. Suppose that the forest language $L$ is defined by a formula $\varphi$ of \Stwol. To say that $X$ describes a forest in $L$ we use the formula  $\varphi^{\beta}$  obtained from $\varphi$ by restricting quantification to nodes satisfying $\beta$.  Note that  $\varphi^\beta$  has free variables $x_1,x_2$ from $\beta$. Since $\beta$ is quantifier-free, the formula $\varphi^\beta$ is also in \Stwol.

Finally, we  say that the part outside $X$  describes a context in $K$. The idea is that the hole of this context corresponds to the set $X$. The logical formula is constructed below. Suppose that $\phi(x)$ is a formula of \Stwol that describes $K$. Note that this formula has a free variable, as with formulas for contexts, which corresponds to the hole. Let $\phi^{\neg \beta}(x)$ be the formula obtained from $\phi(x)$ by restricting quantification to nodes not satisfying $\beta$. The remaining question is: which node should we use for $x$? Any node from $X$ will do, since for each node $y \not \in X$,  all nodes from $X$ have the same relationship to $y$, with respect to the descendant and lexicographic orders. We use the node $x_1$ for $x$.

Summing up, the formula for $KL$ is written below.
\[\exists x_1 \exists x_2 \ \alpha(x_1,x_2)\ \land\  \phi^{\neg \beta}(x_1) \ \land \varphi^\beta.\eqno{\qEd}
\]

\subsection{\texorpdfstring{$\Sigma_2$}{Sigma2} expressions.}
In the proofs, it will sometimes be more convenient to use a type of regular expression instead of formulas. These are called  \emph{$\Sigma_2$ forest expressions} and \emph{$\Sigma_2$ context
  expressions}, and are defined below by mutual recursion:
  \begin{enumerate}[$\bullet$]
  \item Any forest (respectively, context) language that is closed
    under pieces is a $\Sigma_2$ forest (respectively, context)
    expression.  For any label $a \in A$, $\set{a\hole}$ is a
    $\Sigma_2$ context expression. Likewise for $\set{\hole}$, the language containing only the empty context.
\item If $K,K'$ are $\Sigma_2$ context expressions and $L,L'$ are
  $\Sigma_2$ forest expressions, then
  \begin{enumerate}[$-$]
  \item $K \cdot K'$ is a $\Sigma_2$ context expression;
\item $L + L'$ is a $\Sigma_2$ forest expression;
\item $K \cdot L$ is a $\Sigma_2$ forest expression.
\item $L \cup L'$ is a $\Sigma_2$ forest expression. 
\item $K \cup K'$ is a $\Sigma_2$ forest expression. 
\item $L \cap L'$ is a $\Sigma_2$ forest expression. 
\item $K \cap K'$ is a $\Sigma_2$ forest expression. 
  \end{enumerate}
  \end{enumerate}
From Lemmas~\ref{lem:pi1-lemma} and~\ref{fact:simple-lex}  it follows that languages defined by  $\Sigma_2$ forest
and context expressions are definable in \Stwol.
%
%
%

\section{Characterization of \texorpdfstring{\Dtwol}{Dtwol}}
\label{sec:char-inter}
The main result of this paper is the following theorem:

\begin{thm}[Effective characterization of $\Delta_2$ with descendant
  and lexicographic orders]\label{thm:main}\ \vspace{-0.3cm} \\ 
 A forest language is definable in \Dtwol if and only if its syntactic
  forest algebra satisfies the following  identity, called the
  $\Delta_2$ identity,
  \begin{equation}
    \label{eq:swallow}
    v^\omega w v^\omega = v^\omega \qquad \mbox{for }w \piece v\ .
  \end{equation}
\end{thm}

Before we prove the main theorem, we state and prove
an important corollary.
\begin{cor}\label{cor:decid}
  It is decidable whether a forest language can be defined in \Dtwol.
\end{cor}
\begin{proof}
  We assume that the language is represented as a forest algebra. This
  representation can be computed based on other representations, such
  as automata or monadic second-order logic. 

  Once the forest algebra is given, the $\Delta_2$ identity can be
  tested in polynomial time by searching through all elements of the
  algebra. The relation $\piece$ can be computed in polynomial time,
  using a fixpoint algorithm as in~\cite{fo2tree}.
\end{proof}


The following lemma gives the easier implication from the main theorem.

\begin{lem}\label{lemma:correctness}
  Let $\varphi$ be a formula of \Stwol and let $q \piece p$ be two
contexts. For $n \in \Nat$ sufficiently large, forests  satisfying $\varphi$ are closed under replacing
$p^n p^n$ with $p^nqp^n$.
\end{lem}
\begin{proof}
  Assume that $\varphi$ is $\exists x_1 \cdots x_k \forall y_1 \ldots
  y_l \psi$, with $\psi$ quantifier-free.  Any first-order definable
  tree language is aperiodic~\cite{schutzenberger}, i.e.~there is a number $m$ such that any
  context $p^i$ can be replaced by $p^j$ without affecting membership
  in the language, for any $i,j\geq m$. We set $n=2m+(k+1)(l+1)$.

 Consider a forest $t=rp^np^ns$ that satisfies $\varphi$. We want to show that
 the forest $rp^nqp^ns$ also satisfies $\varphi$. By aperiodicity, it
 is sufficient to show that 
 for some numbers $i,j>m$, $rp^iqp^js$ satisfies $\varphi$.

 Because $t$ satisfies $\varphi$ we can fix $k$ nodes $x_1,\ldots, x_k$
 that make $\forall y_1 \ldots y_l \psi$ true. By the choice of $n$,
 $t$ can be decomposed as $rp^ip^lp^js$ such that $i,j\geq m$ and the
 middle $p^l$ part contains none of the nodes $x_1,\ldots,x_k$. We show that
 $t'=rp^ip^lqp^js$ satisfies $\varphi$, which will conclude the proof
 of the lemma.

 We identify the nodes of $t$ with the nodes of $t'$ outside the inserted
 context $q$.  Consider the valuation of the variables $x_1,\cdots, x_k$ that
 we fixed above. We show that this valuation, when seen as nodes of $t'$, makes
 $\varphi$ true. Indeed, for any valuation for the variables $y_1,\cdots,y_l$
 in $t'$, we show that the valuation makes the quantifier-free part $\psi$
 true. This is obvious if none of the $y_i$ are in $q$ because $\varphi$ holds
 in $t$ and the insertion of $q$ does not affect the relationship $<$ and
 $\lex$ between the selected nodes.  If some of the $y_i$ are in $q$ then one
 of the contexts $p$ in the middle block $p^l$ of $t'$ does not contain any
 variable. As removing the context $p$ that does not contain any variable does
 not affect the relationship $<$ and $\lex$ between the selected nodes, $\psi(x_1,\cdots,x_k,y_1,\cdots,y_l)$
 holds on $rp^ip^lqp^js$ iff it holds on $rp^ip^{l-1}qp^js$.
 Moreover, because $q$ is a piece of $p$, replacing $q$ by $p$  does
 not affect the relationship $<$ and $\lex$ between the selected nodes and therefore
$\psi(x_1,\cdots,x_k,y_1,\cdots,y_l)$ holds on $rp^ip^{l-1}qp^js$ iff it holds on $rp^ip^{l-1}pp^js=t$.
 Hence $\psi$ must hold with the new valuation.
\end{proof}

The rest of the paper contains the more difficult implication of
Theorem~\ref{thm:main} which is a consequence of the proposition below.

\begin{prop}\label{prop:bottom-up}
  Fix a morphism $\alpha : \freef A \to (H,V)$, with $(H,V)$
  satisfying the $\Delta_2$ identity.  For each $h \in H$, the set $L_h$
   of forests $t$ with type $\alpha(t)=h$ is definable by a $\Sigma_2$ forest
  expression, and thus also by a formula of \Stwol.
\end{prop}

Before proving this proposition, we show how it concludes the proof of
Theorem~\ref{thm:main}.  Since $\Sigma_2$ expressions allow union, the
above proposition shows that any language recognized by $\alpha$ can
be defined by a $\Sigma_2$ forest expression. In particular, if $L$ is
recognized by $\alpha$, then both $L$ and its complement can be
defined by $\Sigma_2$ forest expressions, and consequently formulas of
\Stwol.  Since the complement of a \Stwol formula is a \Ptwol
formula, we get the right-to-left implication in
Theorem~\ref{thm:main}.

The rest of the paper is devoted to showing
Proposition~\ref{prop:bottom-up}. The proof is by induction on two
parameters.  For the second parameter, we need to define a pre-order on $H$. We say that a type $h$ is
\emph{reachable from} a type $g$ if there is a context type $v \in V$ such
that $h=vg$. If $h$ and $g$ are mutually reachable from each other,
then we write $h \sim g$. Note that $\sim$ is an equivalence relation.
Note also that if $g$ is reachable from $h$, then $h$ is a piece of
$g$. We write $\botcomp$ for the set of types $h$ that can be reached
from every type $g \in H$. Note that $\botcomp$ is not empty, since it
contains the type $h_1 + \cdots + h_n$, for any enumeration
$H=\set{h_1,\ldots,h_n}$.

The proof of Proposition~\ref{prop:bottom-up} is by induction on the
size of the algebra $(H,V)$ and then on the position of $h$ in the
reachability pre-order. The two parameters are ordered
lexicographically, the most important parameter being the size of the
algebra. That is we will either decrease the size of the algebra or
stay within the same algebra, but go from a type $g$ to a type $h$
such that $h$ is reachable from $g$ but not vice versa. As far as $h$
is concerned, the induction corresponds to a bottom-up pass, where
types close to the leaves are treated first.

Part of the  induction proof is presented in Section~\ref{sec:induction-step}.
However, the induction breaks down for types from $\botcomp$, which are
treated in Section~\ref{sec:top-down}.

\section{Types  outside \texorpdfstring{$\botcomp$}{Hbot}}\label{sec:induction-step}

In this section we prove Proposition~\ref{prop:bottom-up} for forest types outside $\botcomp$. We fix such a forest type $h$ for the rest of the section. By induction assumption, for each type $g \not \sim h$ from
which $h$ is reachable, we have a $\Sigma_2$ forest expression defining the
language $L_g$ of forests of type $g$. (The case when there are no
such types $g$ corresponds to the induction base, which is treated the
same way as the induction step.) In this section we assume that $h$ is
outside $\botcomp$, and we will produce a $\Sigma_2$ forest expression
for $L_h$. The case where $h$ is in $\botcomp$ will be treated in
Section~\ref{sec:top-down}.

In the following, we will be using the \emph{stabilizer} of $h$,
defined as
\[
  \stab(h)  =  \set{v : vh \sim h} \subseteq V \ .
\]
We say that a context type $v$ stabilizes $h$ if it belongs to the
stabilizer of $h$.  The key lemma is that the $\Delta_2$ identity
implies that the stabilizer is a submonoid of $V$.

\begin{lem}\label{lemma:closed-under-comp}
  The stabilizer of $h$ only depends on the $\sim$-class of $h$. In
  particular, it is a  submonoid of $V$.
\end{lem}
\begin{proof}
  We need to show that if $h \sim h'$ then
  $\stab(h)=\stab(h')$. Assume $v \in \stab(h)$. Then $vh \sim
  h$. Hence we have $u_1,u_2,u_3$ such that $h=u_1vh, h=u_2h'$ and
  $h'=u_3h$. This implies that $h'=u_3u_1vu_2 h'$ and therefore
  $h'=(u_3u_1vu_2)^\omega h'$. From the $\Delta_2$ identity we have that
  \[
h'=(u_3u_1vu_2)^\omega h'= (u_3u_1vu_2)^\omega v (u_3u_1vu_2)^\omega h' = (u_3u_1vu_2)^\omega v h'\ .
\]
  Hence $h'$ is reachable from $vh'$. Since $vh'$ is clearly reachable
from $h'$, we get $vh' \sim h'$  and $v \in \stab(h')$.

To see that $\stab(h)$ is a submonoid consider $v,v' \in \stab(h)$. We need to
show that $vv' \in \stab(h)$. Let $h'=v'h$. Because $v' \in \stab(h)$ we have $h'
\sim h$. As $v \in \stab(h)=\stab(h')$ we have $vh'\sim h'$ and hence $vv'h
\sim h$.
\end{proof}

Recall now the piece order on forest types from the end of
Section~\ref{sec:forest-algebra}, which corresponds to removing nodes
from a forest. We say a set $F \subseteq H$ of forest types is
\emph{closed under pieces} if any piece of a forest type $f \in F$
also belongs to $F$. A similar definition is also given for sets of
context types. Another consequence of the $\Delta_2$ identity is:

\begin{lem}\label{lemma:closed-under-piece}
  Each stabilizer is closed under pieces.
\end{lem}
\begin{proof}
  We need to show that if $u$ stabilizes $h$, then each piece $u'$ of
  $u$ also stabilizes $h$. By definition of the stabilizer we have a
  context type $v$ such that $h=vuh$. We are looking for a context type $w$ such
  that $wu'h=h$. From $h=vuh$ we get $h=(vu)^\omega h$. Hence by the
  $\Delta_2$ identity we have $h=(vu)^\omega u' (vu)^\omega h =
  (vu)^\omega u' h$ as desired.
\end{proof}

We now consider two possible cases: either $h+h \sim h$, or not. Equivalently,
we could have asked if $h+\hole$ stabilizes $h$. Again equivalently, we could
have asked if $\hole+h$ stabilizes $h$. When $h+h \sim h$, we will conclude by
induction on the size of the algebra. When $h+h \not \sim h$, we will conclude
by induction on the reachability pre-order. These cases are treated separately
in Sections~\ref{sec:h-does-not-preserve-itself}
and~\ref{sec:h-preserves-itself}, respectively.

\subsection{\texorpdfstring{$h + h \sim h$}{h+h sim h}}
\label{sec:h-preserves-itself}
Let $G$ be the set of pieces of $h$. By assumption that $h+h \sim h$,
we know that both $h + \hole$ and $\hole + h$ stabilize $h$. 
\begin{lem}\label{lem:sub-forest-algebra}
  If $h + h \sim h$ then $(G,\stab(h))$ is a forest
  algebra.
\end{lem}
\begin{proof}
    We need to show that the two sets are closed under all operations:
\begin{eqnarray*}
    stab(h) stab(h) & \subseteq &\stab(h) \\
    G +  G &  \subseteq & G \\
 \hole +G, G + \hole     & \subseteq & \stab(h) \\
    \stab(h)  G & \subseteq & G
\end{eqnarray*}
  The first of the above inclusions follows from
  Lemma~\ref{lemma:closed-under-comp}. For the second inclusion, we
  note that $h+h$ is a piece of $h$ by assumption on $h+h \sim h$. In
  particular, each forest type in $G + G$ is a piece of $h$. For the
  third inclusion, $h+h \sim h$ implies that both $h + \hole$ and
  $\hole + h$ stabilize $h$. Since by
  Lemma~\ref{lemma:closed-under-piece} the stabilizer is closed under
  pieces, we get the third inclusion. For the last inclusion, consider
  $v \in \stab(h)$ and $g \in G$. We need to show that $vg \in
  G$. This holds because $vg$ is a piece of $vh$, which is a piece of
  $h$ as $vh \sim h$.
\end{proof}

Recall that in this section we are dealing with the case when $h$ is
outside $\botcomp$, i.e.~there are some forest types that can be
reached from $h$ but not vice versa. In this case we show that $G$ is a proper
subset of $H$. To see this, we show that for $h_\bot \in \botcomp$, $h_\bot
\not\in G$. Assume for contradiction that $h_\bot \in G$. Then $h_\bot$ is a
piece of $h$ and because $h + \hole \in \stab(h)$ and $\stab(h)$ is closed
under pieces (Lemma~\ref{lemma:closed-under-piece}), we infer $h_\bot + \hole
\in \stab(h)$ and $h$ is reachable from $h_\bot$, a contradiction.

Therefore the algebra from the above lemma is a
proper subalgebra of the original $(H,V)$.  Furthermore, this algebra
contains all pieces of $h$; so it still recognizes the language $L_h$;
at least as long as the alphabet in the morphism is reduced to include
only letters that can appear in $h$.  We can then use the induction
assumption on the smaller algebra to get a $\Sigma_2$ forest
expression for $L_h$.

\subsection{\texorpdfstring{$h + h \not \sim h$}{h+h notsim h}}
\label{sec:h-does-not-preserve-itself}

For $v \in V$, we write $K_v$ for the set of contexts of type $v$. For
$g \in H$, we write $L_g$ for the set of forests of type $g$, and
$M_g$ for the set of trees of type $g$.

Let $G$ be the set of forest types $g$ such that $h$ is reachable from
$g$ but not vice-versa.  By induction assumption,  for each $g \in G$,
the  language $L_g$
is definable by a $\Sigma_2$ forest expression. Our
goal is to give a $\Sigma_2$ forest expression for $L_h$.

\begin{lem}\label{lem:decompo-nonminim}
  Any forest $t$ of type $h$ can be decomposed as $t=ps$, with $s$ a
  forest whose type is  reachable from $h$, and which furthermore is:
\begin{enumerate}[\em(1)]
\item A tree $s=as'$ with the type of $s'$ in $G$; or
\item A forest $s=s_1+s_2$ with the types of $s_1,s_2$ in $G$.
\end{enumerate}
\end{lem}
\begin{proof}
  Consider decompositions of $t$ as $t=ps$. Among such decompositions,
  take a decomposition where the forest $s$ has a type reachable from
  $h$, but $s$ has no subforest with a type reachable from $h$. Such  a
  decomposition always exists as $t$ is of type $h$. If $s$
  is a tree, we get case (1), if $s$ is a forest, we get case (2).
\end{proof}

Note that in the above lemma, the type of the context $p$ must
stabilize $h$, since both the type of $s$ and the type of the whole
forest $t$ are in the class of $h$. Therefore, thanks to the above
lemma, the set $L_h$ of forests with type $h$ can be decomposed as
\begin{equation}
  \label{eq:lh-nonmaximal-decomp}
    L_h = \bigcup_{u \in \stab(h)}\qquad \big(\bigcup_{\substack{a \in A, g \in G\\
      uag=h}} K_uaL_g \qquad \cup \qquad  \bigcup_{\substack{g_1,g_2 \in G\\
      u(g_1+g_2)=h}} K_u(L_{g_1} + L_{g_2}) \big)
\end{equation}
Note that $L_g$, $L_{g_1}$ and $L_{g_2}$ can all be written as
$\Sigma_2$ forest expressions thanks to the induction assumption from
Proposition~\ref{prop:bottom-up}.  The only thing that remains is
showing that the context language $K_u$ can be defined by a $\Sigma_2$
context expression.  For this, we use the following proposition.

\begin{prop}\label{prop:context-expression}
  For any $v \in V$, the language $K_v$ of contexts of type $v$ is
  defined by a finite union of concatenations of the form $K_{1} \cdots
  K_{m}$ where each context language $K_{i}$ is either:
  \begin{enumerate}[\em(1)]
  \item A singleton language $\set{a \hole}$ for some $a \in A$; or
\item A context language closed under pieces; or
\item  A  context  language    $\hole + M_g$
  or $M_g + \hole$ for  some $g \in H$.
  \end{enumerate}
\end{prop}

The proof of this proposition will be presented in
Section~\ref{sec:treat-cont-like-words}. Meanwhile, we show how the
proposition gives a $\Sigma_2$ context expression for each language $K_u$
in~(\ref{eq:lh-nonmaximal-decomp}). The singleton languages, and the
languages closed under pieces are $\Sigma_2$ context expressions by
definition.  The only potential problem is with the languages $\hole +
M_g$ or $M_g + \hole$ that appear in the proposition.  Since the
context types $u$ that appear in~(\ref{eq:lh-nonmaximal-decomp})
stabilize $h$, the forest type $g$ has to be such that $\hole + g$ or
$g + \hole$ stabilizes $h$. In either case, $g$ cannot be reachable
from $h$, since $h + h \not \sim h$ and the stabilizer is closed under
pieces. As $h$ is obviously reachable from $g$, the language $L_g$ is
definable by a $\Sigma_2$ forest expression thanks to the induction
assumption from Proposition~\ref{prop:bottom-up}. Finally, $M_g$ is
the intersection of $L_g$ with the set of all trees, which is
definable by a $\Sigma_2$ forest expression.

\section{Treating contexts like words}
\label{sec:treat-cont-like-words}
In this section, we prove Proposition~\ref{prop:context-expression}.
The basic idea is that a context is treated as a word, whose letters are smaller contexts.
The proof strategy is as follows. First, in Section~\ref{sec:inter-words}, we
present the characterization of \Dtwo for words, which was shown by Pin and
Weil in~\cite{weilpinpoly}. This characterization is slightly strengthened to
include what we call stratified monoids, which are used to model the contexts
that appear in Proposition~\ref{prop:context-expression}.  Then, in
Section~\ref{sec:proof-prop-refpr}, we apply the word result, in its
strengthened form, to prove Proposition~\ref{prop:context-expression}.

\subsection{\texorpdfstring{\Dtwo}{Dtwol} for words}
\label{sec:inter-words}
In this section we present the characterization of \Dtwo for words,
extended to stratified monoids.  A \emph{stratified monoid} is a
monoid $M$ along with a pre-order $\preceq$ that satisfies the following
property:
\[
  m^\omega n m^\omega = m^\omega \qquad \mbox{for }n \piece m\ .
\]
A subset $N \subseteq M$ is called \emph{downward closed under $\piece$} if for every $n \in N$, and every $m \piece n$, we also have $m \in N$.
 
\begin{prop}\label{prop:word-da}
  Let $A$ be an alphabet (possibly infinite), and let $\beta : A^* \to
  M$ be a morphism into a stratified monoid $(M, \piece)$ that satisfies
  the identity
  \begin{equation}
    \label{eq:da}
        (mn)^\omega m (mn)^\omega = (mn)^\omega
  \end{equation}
  For any $m \in M$, the language $\beta^{-1}(m)$ is defined by a
  finite union of expressions
  \[
    A_0^* B_1 A_1^* \cdots B_{i}A_i^*
  \]
  where each $B_j$ is of the form $A \cap \beta^{-1}(n)$ for some $n
  \in M$, and each $A_j$ is of the form $A \cap \beta^{-1}(N)$ for
  some $N \subseteq M$ downward closed under $\piece$.
\end{prop}

The difference between the above result and the main technical result
in Pin and Weil is twofold. First, we use infinite alphabets here.
Second, we use stratified monoids to get a
stronger conclusion, where the letters in the blocks $A_i^*$ are
downward closed.  Both differences are necessary for our application
to context languages.

Our proof is a straightforward adaptation of a proof of Th\'erien and Wilke
in~\cite{therienwilkefo2}, which analyzed the languages recognized by
semigroups in DA. 

  Before proving this result, we remark how
  Proposition~\ref{prop:word-da} gives the
characterization of \Dtwo presented by Pin and Weil:
\begin{cor}[Pin and Weil~\cite{weilpinpoly}]
  A word language (over a finite alphabet) is definable \Dtwo if and only if
  its syntactic monoid satisfies the identity~(\ref{eq:da}).
\end{cor}
\begin{proof}
  The only if implication is shown using a standard
  Ehrenfeucht-Fraïssé argument, we only consider the if implication.

  Let then $L\subseteq A^*$ be a language recognized by a morphism
  $\beta : A^* \to M$, with $M$ satisfying (\ref{eq:da}). We can see
  this $M$ as a stratified monoid under the identity pre-order. By
  applying Proposition~\ref{prop:word-da}, we see that each inverse
  image $\beta(m)$ is defined by an expression as in
  Proposition~\ref{prop:word-da} (the downward closure is a vacuous
  condition, since the order is trivial). Since each such expression
  is clearly expressible in \Stwo, we get that $L$ is definable
  in \Stwo. Furthermore, by Proposition~\ref{prop:word-da} also
  the complement of $L$ is definable in \Stwo, and therefore $L$
  is also definable in \Ptwo.
\end{proof}

We now proceed to the proof of Proposition~\ref{prop:word-da}.  The proof is by
induction on the size of $\beta(A) \subseteq M$ or, equivalently, the number of
elements in the monoid that correspond to single letters. In the proof
of Thérien and Wilke, the induction was simply on the size of the alphabet,
but this will not work here, since the alphabet is infinite.

 We use the term
\emph{$\Sigma_2$ word expression} for the word expressions as in the
statement of the Proposition~\ref{prop:word-da}. It is not difficult to
show that languages defined by $\Sigma_2$ word expressions are closed under
union, intersection and concatenation.

We will use the following notation. Given two elements $m$ and $n$ of $M$ we
say that $m \lgreensim n$ if there exist $k,l \in M$ such that $m=kn$ and $n=lm$. We
say that $m \rgreensim n$ if there exist $k,l \in M$ such that $m=nk$
and $n=ml$. These are the left and right Green's relations.

A classical consequence of aperiodicity, itself a consequence
of~(\ref{eq:da}), is:
\begin{equation}\label{eq:wilke2}
  m \rgreensim n\  \land \  m \lgreensim n \ \Rightarrow \ m=n\ 
  \qquad \mbox{ for }m,n \in M\ .
\end{equation}

  We will also use the following property of monoids
  satisfying~(\ref{eq:da}), which can be proved along the lines of
  Lemma~\ref{lemma:closed-under-comp}.
\begin{equation}\label{eq:wilke1}
 m \rgreensim n \rgreensim mk \ \Rightarrow nk \rgreensim n \qquad \mbox{ for
  }m,n,k \in M\ .
\end{equation}

\begin{lem}\label{lemma:wilke3}
  For all $m \in M$, the language $U_m=\set{w : m \beta(w) \rgreensim m}$ is
  definable by a $\Sigma_2$ word expression. 
\end{lem}
\begin{proof}
  Let $A_m$ be the set of letters $a$ of $A$ such that $m\beta(a) \rgreensim
  m$. In other words, $A_m = A \cap \beta^{-1}(N)$, where $N$ is the set
  $\set{n: mn\rgreensim m}$.

  We will show that $U_m = A_m^*$. Stated differently, a word belongs $U_m$ if
  and only if all of its letters belong to $A_m$.

  Thanks to~(\ref{eq:wilke1}), for all $n \rgreensim m$ we have $n\beta(a)
  \rgreensim m$. Hence by induction on the length of $w \in A_m^*$ we can prove
  $w \in U_m$. For the converse implication, let $w$ be a word outside $A_m^*$,
  of the form $w=w_1aw_2$ with $w_1 \in A_m^*$ and $a \not\in A_m$. Then
$m\beta(w) = n \beta(a) n'$ for some $n,n'$, and from the discussion above we
have $n \rgreensim m$. Hence by~(\ref{eq:wilke1}) and by hypothesis on $a$, $n
\beta(a) \not\rgreensim m$ and $w$ cannot be in $U_m$.

 To conclude, we need to show that $N$ is closed under $\piece$.
 Indeed, let $k \piece n$ and let $n \in N$. By assumption on the
 monoid being stratified, we have $n^\omega k n^\omega = n^\omega$. In
 particular, we have 
 \[
   mn^\omega k n^\omega \rgreensim mn^\omega \rgreensim m\ .
 \]
 From the above it follows that $mn^\omega k \rgreensim m$, which gives
 $mk \rgreensim m$ by~(\ref{eq:wilke1}), and hence $k \in N$.
\end{proof}

Note that by~(\ref{eq:wilke1}) we have $U_n=U_m$ whenever $m \rgreensim n$.

\begin{lem}\label{lemma:wilke4}
  For any $m \in M$, the following language can be defined by a
  $\Sigma_2$ word expression:
\[
V_m =  \set{a_1\cdots a_i: \ \beta(a_1 \cdots a_i)=m\mbox{ and }
  \beta(a_1 \cdots a_{i-1}) \not\rgreensim  m}\,.
\]
\end{lem}

Before we give the proof, we show how it concludes the proof of
Proposition~\ref{prop:word-da}. Consider the set  $\set{w : \beta(w) \rgreensim
  m}$. Each word $w$ in this set can be written as $uv$ where $u$ is the
smallest prefix of $w$ such that $\beta(u) \rgreensim m$. Hence we have:
\[
   \set{w : \beta(w) \rgreensim m} = \bigcup_{n : n \rgreensim\,m} V_n U_m\,.
\]
Since $\Sigma_2$ word expressions are closed under union and concatenation, the
above language is definable by a $\Sigma_2$ word expression  thanks to
Lemmas~\ref{lemma:wilke3} and \ref{lemma:wilke4}.  Using a symmetric version of
Lemma~\ref{lemma:wilke3} and Lemma~\ref{lemma:wilke4} for $\lgreensim$, we can get
an $\Sigma_2$ word expression for $\set{w : \beta(w) \lgreensim m}$. But we also know
from~(\ref{eq:wilke2}) that
\[
  \beta^{-1}(m)= \set{w : \beta(w) \rgreensim m} \cap \set{w : \beta(w) \lgreensim m}
\]
and the result follows by closure of $\Sigma_2$ word expressions under
intersection.

\medskip
\begin{proof}[Proof of Lemma~~\ref{lemma:wilke4}]
  We say that $m$ is \emph{a prefix} of $n$ if there exists $k \in M$
  such that $n=mk$. This defines a pre-order in $M$. The proof is by
  induction on the position of $m$ relative to this pre-order.

  The induction base is when $m$ has no proper prefixes: If $m=nk$ then $n \rgreensim m$.
  In this case the language $V_m$ contains at most the empty word, since the
  condition on $a_1 \cdots a_{i-1}$ is infeasible.  Clearly both languages
  $\emptyset$ and $\set \epsilon$ are $\Sigma_2$ word expressions.

Assume now that $m$ is not minimal. Each word $w$ of $V_m$ can be written as
$ua$ where $a \in A$, $\beta(u)=k$ and $k\beta(a)=m$. Furthermore, $u$ can be
written as $u_1u_2$ where $u_1$ is the smallest prefix of $u$ such that
$n=\beta(u_1) \rgreensim k$. We therefore have:
\[
  V_m  =  \bigcup  V_n U_{n,k} a \qquad \mbox{with }  U_{n,k}=\set{ w \in A^* :
    n\beta(w)=k}
\]
where the union is taken for $n,k \in M$ such that $n \rgreensim k$,
$k \not \rgreensim m$ a prefix of $m$ and for $a \in A$ with $k
\beta(a)=m$. By induction the language $V_n$ is definable by a
$\Sigma_2$ word expression. It is also clear that $U_{n,k} \subseteq
U_n$. Recall from the proof of Lemma~\ref{lemma:wilke3} that
$U_n=A_n^*$ where $A_n=A \cap \beta^{-1}(N)$ for some $N \subseteq M$.
Therefore we also have $\beta(U_{n,k}) \subseteq A_n^*$.
From~(\ref{eq:wilke1}) and the fact that $k\beta(a)=m$ we know that $a
\not\in A_n$, therefore $A_n$ is a proper subset of $A$. Let $\beta'$
be the restriction of $\beta$ to $A_n$. We have
$U_{n,k}=\bigcup_{nx=k}\beta'^{-1}(x)$, from induction on the size of
the alphabet in Proposition~\ref{prop:word-da} we obtain a $\Sigma_2$
word expression for $U_{n,k}$. This concludes the proof of this lemma
as $\Sigma_2$ word expressions are closed under concatenation and union.
\end{proof}

\subsection{Proof of Proposition~\ref{prop:context-expression}}
\label{sec:proof-prop-refpr}
We now proceed to show how the word result stated in
Proposition~\ref{prop:word-da} can be lifted to the context result in
Proposition~\ref{prop:context-expression}. Proposition~\ref{prop:context-expression}
says that for any context type $v \in V$, the set of contexts of type $v$
  is described by a finite union of expressions of the form $K_{1} \cdots
K_{m}$ where each $K_{i}$ is either: a singleton language $\set{a \hole}$; or
closed under pieces; or an expression $\hole + M_g$ or $M_g + \hole$.

The basic idea is that we treat the context as a word over an infinite
alphabet, which we call $B$. This alphabet has two kinds of letters.
Both kinds are contexts:
\begin{enumerate}[$\bullet$]
\item Contexts of the form $a \hole$, for $a \in A$.
\item Contexts of the form $t + \hole$ or $\hole + t$, for $t$ a tree
  over $A$.
\end{enumerate}

Consider now the morphism $\beta : B^* \to V$, which is simply $\alpha$
restricted to the contexts in~$B^*$.  Every context $p$ in $K_v$ can be
decomposed as $p=b_1 \cdots b_m \in B^*$. In particular, we have
\[
  K_v =\beta^{-1}(v)\ .
\]

We can treat $V$ as a stratified monoid, by using the piece relation
$\piece$ as the pre-order. By applying Proposition~\ref{prop:word-da}, we
see that the inverse image $\beta^{-1}(v)$ can be presented as a
finite union of expressions of the form:
  \[
     (B \cap \beta^{-1}(N_0))^*  (B \cap \beta^{-1}(n_1)) \cdots 
    (B \cap \beta^{-1}(n_k))  (B \cap \beta^{-1}(N_k))^*\ ,
  \]
  where $n_1,\ldots,n_k$ are elements of $V$, and $N_1,\ldots,N_k$ are
  a subsets of $V$ that are downward closed under $\piece$.

  We need to show that the expressions used above are of the three
  forms allowed by Proposition~\ref{prop:context-expression}.
  Consider first an expression $(B \cap \beta^{-1}(W))^*$, where $W
  \subseteq V$ is closed under pieces. Since $W$ is closed under
  pieces (as a set of context types), then so is the language $(B \cap
  \beta^{-1}(W))^*$ (as a set of contexts).  Consider next an
  expression of the form $B \cap \beta^{-1}(n)$. This context language
  is a union of languages of the first (singleton) and third ($\hole +
  M_g$ or $M_g + \hole$) types described in
  Proposition~\ref{prop:context-expression}.  The union is not a
  problem for a $\Sigma_2$ word expression, since union distributes
  across concatenation.

\section{Types in \texorpdfstring{$\botcomp$}{Hbot}}
\label{sec:top-down}

\def\ordereq{\sim}
\def\order{\leq}

Recall that in Section~\ref{sec:induction-step}, we managed to find a
$\Sigma_2$ forest expression for each set $L_h$, assuming $h$ was
outside $\botcomp$. Our techniques failed for forest types $h \in
\botcomp$, i.e.~forest types reachable from every other forest type.
In this section, we deal with these forest types.

In order to deal with the types from $\botcomp$, we will have to do a different
induction, this time on context types.  This induction, stated in 
Proposition~\ref{prop:contexts}, is expressed in terms of an equivalence
relation $\equiv_v$.  Given a context type $v \in V$ and two forest
types $h,h'$, we write 
\begin{equation}
  \label{eq:equiv-forest}
  h \equiv_v h' \qquad \mbox{if} \quad \forall u \in V \ vuh = vuh'\ .
\end{equation}
We extend this equivalence relation to context types, by
\begin{equation}
  \label{eq:equiv-context}
  w \equiv_v w' \qquad \mbox{if} \quad \forall u \in V\  \forall h \in
  H\ vuwh = vuw'h\ .
\end{equation}
By abuse of notation, we also lift the equivalence relation $\equiv_v$
to forests, considering two forests $s,t$ equivalent when their forest
types are equivalent. It is this meaning that is used in the statement
below.
\begin{prop}\label{prop:contexts}
  For any context type $v$, every equivalence class of forests under
  $\equiv_v$ is forest language definable by a $\Sigma_2$ forest
  expression.
\end{prop}

From Proposition~\ref{prop:contexts} we immediately obtain a
$\Sigma_2$ forest expression for $L_h$, as the equivalence class of
$\equiv_v$ containing $h$, where $v=\hole$.  Hence the proof of
Proposition~\ref{prop:contexts} ends the proof of
Proposition~\ref{prop:bottom-up}.  We note that the proof of
Proposition~\ref{prop:contexts} will be using the $\Sigma_2$ forest
expressions $L_h$ for types $h \not\in \botcomp$ that have been developed
in Section~\ref{sec:induction-step}. In particular if an equivalence
class of $\equiv_v$ is contained in $H \setminus \botcomp$, then it can easily
be defined by the disjunction of all the $\Sigma_2$ forest expressions
corresponding to each type. The difficulty is to handle equivalence
classes that intersect $\botcomp$.

The rest  of Section~\ref{sec:top-down} is devoted to proving Proposition~\ref{prop:contexts}.
The proof uses the following pre-order on context types. We say that a
context type $u$ is a \emph{prefix} of a context type $v$ if there exists a context type $w$
such that $v=uw$ (we also say that $v$ is an \emph{extension} of
$u$). We overload the use of
$\sim$ and denote by $\ordereq$ the equivalence relation induced by
the prefix pre-order, i.e.~$v \sim w$ holds if $v$ is both a prefix
and an extension of $w$.
The proof of Proposition~\ref{prop:contexts} is by induction on the
position of 
$v$ in the prefix pre-order,  starting with context types that have no
proper extension, and ending at the context type  $v=\hole$ that has no
proper prefix.

\subsection{ The induction base}
The base of the induction in the proof of Proposition~\ref{prop:contexts} is
when the context type $v$ has no proper extension, i.e.~$v \sim w$ holds for
all extensions $w$ of $v$. We will show that such a context type is necessarily
\emph{constant}, i.e.~$vg=vh$ holds for all forest types $g,h \in H$.  This
gives the induction base, since for a constant context type $v$, there is only one
equivalence class of $\equiv_v$, and this class is, by definition, the set of all
forests, which can be defined by a $\Sigma_2$ forest expression (it is closed under pieces).

\begin{lem}\label{lemma:minimal-then-constant}
  A context type has no proper extension if and only if it is constant.
\end{lem}

The if direction is immediate: if a context type $v$ is constant, then
$vu=v$ holds for all context types $u$, and therefore $v$ has no
proper extension. For the converse implication, as well as in the rest
of Section~\ref{sec:top-down}, we will use the notion of stabilizers
for context types:
\[
  \stab(v)  =  \set{w : vw \ordereq v}\ .
\]
When $u \in stab(v)$, we say that \emph{$u$ stabilizes $v$}.  As for
stabilizers of forest types (recall Lemma~\ref{lemma:closed-under-comp} and
Lemma~\ref{lemma:closed-under-piece}), the $\Delta_2$ identity implies that the
stabilizer $\stab(v)$ is a submonoid of $V$ and it is closed under pieces.  The
following lemma implies Lemma~\ref{lemma:minimal-then-constant}, since its
assumptions are met by a context type without proper extensions.
Recall that
$\botcomp$ is the equivalence class of $\sim$ that contains all types reachable
from any other type.


\begin{lem}\label{lemma:at-least-one-disjoint}
  If both $\botcomp+ \hole$ and $\hole + \botcomp$ intersect $\stab(v)$, then
  the context type $v$ is constant.
\end{lem}
\proof
  First note that if some context type in $\botcomp + \hole$
  stabilizes $v$, then all context types in $\botcomp + \hole$
  stabilize $v$, likewise for $\hole +\botcomp$. This is because the
  stabilizer is closed under pieces, and every type in $H + \hole$ is
  a piece of every type in $\botcomp + \hole$.

 Let $f = h_1 + \cdots +
  h_n$, for some arbitrary enumeration $h_1,\ldots,h_n$ of $H$. As we
  noted above, both $f + \hole$ and $\hole+ f$ stabilize $v$. Therefore,
  \[
    v( \omega f + \hole  + \omega f)  \sim  v\ .
  \]
  The context type $( \omega f + \hole + \omega f)$ is constant,
  since any forest type $h \in H$ is a piece of $f$, and therefore
  by~(\ref{eq:swallow}) we have
  \[
    \omega f + h + \omega f = (f + \hole)^\omega \cdot (h + \hole)
    \cdot (f + \hole)^\omega \cdot  0  = (f + \hole)^\omega 
    \cdot (f + \hole)^\omega \cdot  0 = \omega  f + \omega f \ .
  \]
Hence the context type $v$ is constant as it is equal to $wv( \omega f
+ \hole + \omega f)$ for some context type $w$.\qed

\subsection{The induction step}
We now proceed to the induction step in
Proposition~\ref{prop:contexts}.  Recall that our goal is to find a
$\Sigma_2$ forest expression for every equivalence class of
$\equiv_v$. For a forest language $L$, we denote by $[L]_v$ the union
of equivalence classes of $\equiv_v$ that intersect $L$, i.e.
\[
  [L]_v =  \set{ t : t \equiv_v s \mbox{ for some }s \in L} \quad
  \supseteq L\ .
\]
We use a similar notation $[K]_v$ for languages of contexts.  We say
that a forest language is a \emph{$v$-overapproximation}
of a forest language $L$ if it contains $L$, but is contained in
$[L]_v$. In other words, a $v$-overapproximation may add forests to
$L$, but it adds no new forest types, at least as far as the context
type $v$ is concerned.  Note that a language may have several
$v$-overapproximations.
\begin{prop}\label{prop:overapprox}
  For every $h \in H$, some $v$-overapproximation of $L_h$ can be
  defined by a $\Sigma_2$ forest expression.
\end{prop}

\newcommand{\overap}[1]{\hat{#1}} The above result concludes
the proof of Proposition~\ref{prop:contexts}. To see this, consider an equivalence
class consisting of forest types $f_1,\cdots,f_n$. The set of
forests with a type in the class is by definition equal to
$\bigcup L_{f_i}$. By definition of $v$-overapproximation this
set is also equal to $\bigcup \overap L_{f_i}$ where $\overap
L_{f_i}$ is any $v$-overapproximation of $L_{f_i}$. Hence it is definable by a
$\Sigma_2$ forest expression by Proposition~\ref{prop:overapprox}.

The rest of this section is devoted to proving
Proposition~\ref{prop:overapprox}. 

When $h$ is outside $\botcomp$, then $L_h$ itself, which is its own
$v$-overapproximation, is definable by a $\Sigma_2$ forest expression
by the results from the previous sections. The problem is when $h$ is
in $\botcomp$.  Because $v$ has some proper extension, from
Lemma~\ref{lemma:at-least-one-disjoint} we know that at least one of  $\hole +
\botcomp$ or $\botcomp + \hole$ is disjoint with the stabilizer of $
v$. Without loss of generality we assume
\begin{equation}
  \label{eq:assumption-left}
  \hole + \botcomp \  \cap\  \stab(v) \quad  = \quad 
  \emptyset\ .
\end{equation}
In other words, if the type of a context stabilizes $v$, then it is
possible that some tree to the left of the hole has a type in $\botcomp$,
however all trees to the right of the hole must have types outside $\botcomp$.

In order to obtain a $v$-overapproximation of $L_h$ for $h \in \botcomp$ we
will use the following decomposition of forests with types in $\botcomp$.

\begin{lem}\label{lemma:decomp}
  Any forest $t$ of type in $\botcomp$ has a decomposition $t=ps$  where $p$ is a context whose type stabilizes $v$, and  $s$ is a forest of type   in  $\botcomp$ that has one of the two forms below.
  \begin{enumerate}[\em(1)]
\item\label{item:decomp2bis} $s=as'$ with  the type of the forest $s'$ outside $\botcomp$; or
	\item\label{item:decomp2} $s=as'$ with  the type of the context $a\hole$ not stabilizing
	    $v$; or
	\item\label{item:decomp3bis} $s=s_1+s_2$ with the types of the forests $s_1,s_2$ outside $\botcomp$; or
  \item\label{item:decomp3} $s=s_1+s_2$  and \begin{enumerate}[$\bullet$]
	\item if $s_1$ has  type in $\botcomp$, then the type of $\hole + s_2$  does not stabilize $v$.
		\item if $s_2$ has  type in $\botcomp$, then the type of $s_1 + \hole$  does not stabilize $v$.
\end{enumerate}
  \end{enumerate}
\end{lem}
\begin{proof}
Consider the set $D$ of all possible pairs $(p,s)$ such that $t=ps$, the type of $p$ preserves $v$ and the type of $s$ is in $\botcomp$.  Take a pair $(p,s) \in D$ that is maximal in the following sense: if $q$ is a nonempty context, then $D$ has no pair with $pq$ on the first coordinate.

Suppose $s$ is a tree of the form $s=as'$.  If $s'$ has a type outside
$\botcomp$, we get item~(\ref{item:decomp2bis}). If $s'$ has type in
$\botcomp$, then by maximality, the context type of $pa\hole$ does not stabilize $v$,
and therefore $a \hole$ has a type that does not stabilize $v$, so we get
item~(\ref{item:decomp2}).

Suppose   $s$ is a forest of at least two trees.
Consider any partition of $s$ into two nonempty forests $s=s_1+s_2$.
If both $s_1,s_2$  have type outside $\botcomp$, then we get item~(\ref{item:decomp3bis}). Otherwise, we get case~(\ref{item:decomp3}) by maximality of $(p,s)$.
\end{proof}

From Lemma~\ref{lemma:decomp}, we have for $h \in \botcomp$:
\[
 L_h = \bigcup_{\substack{u \in \stab(v)} } \bigcup_{\substack{f
       \in \botcomp  \\ uf=h}} K_u \cdot Y_f
\]
where $Y_f$ stands for the set of all forests $s$ that have type $f \in \botcomp$ and that
satisfy one of the conditions~(\ref{item:decomp2bis})-(\ref{item:decomp3}) of
Lemma~\ref{lemma:decomp}.
To get the $v$-overapproximation of $L_h$ we will use
$v$-overapproximations for the smaller expressions above, as stated by the
following two lemmas. The first lemma is concerned with languages of the form $Y_f$.

\begin{lem}\label{lemma:y-approx}
  For every $f \in \botcomp$, some $v$-overapproximation $\overap Y_f$ of $Y_f$ can be
  defined by a $\Sigma_2$ forest expression.
\end{lem}

For the second lemma, concerning $K_u$, we need a more careful
statement. The overapproximation that we give is not really an
overapproximation of $K_u$, but it is an overapproximation that works
as long as a forest of type in $\botcomp$ is inserted into the hole.
\begin{lem}\label{lemma:k-approx}
  For any $u \in stab(v)$, one can define a $\Sigma_2$ context
  expression $\overap K_u$ such that for any forest $s$ of type in
  $\botcomp$, $\overap K_us$ is a $v$-overapproximation of $K_us$.
\end{lem}

These two lemmas are proved in Sections~\ref{sec:proof-lemma-y}
and~\ref{sec:proof-lemma-k}, respectively. First we show how they
complete the proof of Proposition~\ref{prop:overapprox}. We write $L_\bot$ for the
set of all forests with a type in $\botcomp$. We claim that the  following
language
\[
 \overap L_h =   \bigcup_{\substack{u \in \stab(v)} } \bigcup_{\substack{f
     \in \botcomp  \\ uf=h}} \overap K_u \cdot (\overap Y_f \cap L_\bot)
\]
is a $v$-overapproximation of $L_h$. 

The first property required from a $v$-overapproximation, $L_h \subseteq
\overap L_h$, is immediate. For the second part, $\overap L_h \subseteq
[L_h]_v$, we need a bit more effort. We show a stronger result, namely that for
any $u$ and $f$ as in the summation above, we have
\[
\overap K_u (\overap Y_f  \cap
 L_\bot) \quad   \subseteq \quad   [K_u Y_f]_v
\]
This completes the proof of $\overap L_h \subseteq [L_h]_v$, since
$[\_]_v$ distributes across union.  To prove the above, we apply the
properties of  $\overap K_u$ and $\overap Y_f$ to get
\[
 \overap K_u (\overap Y_f \cap L_\bot)
\quad  \subseteq \quad   [K_u  \cdot  (\overap Y_f   \cap L_\bot)]_v
\quad  \subseteq \quad   [K_u  \cdot \overap Y_f   ]_v
\quad  \subseteq \quad   [K_u  \cdot [Y_f]_v   ]_v
\]
To complete the proof, we would like to replace $[Y_f]_v$ by $Y_f$ in
the last expression above. This can be done thanks to the following
easily verifiable consequence of the fact that $\equiv_v$ is a congruence for
forest algebras.
\begin{fact}
 For any set of contexts $K$ and set of languages $L$, we have
 $[K[L]_v]_v=[KL]_v$.
\end{fact}

Thanks to Lemmas~\ref{lemma:k-approx} and~\ref{lemma:y-approx}, the
only thing keeping $\overap L_h $ from being defined by a $\Sigma_2$
forest expression is  the language $L_\bot$. We deal with this
language in the following lemma.

\begin{lem}
 The language  $L_\bot$   is definable by a
 $\Sigma_2$ forest expression.
\end{lem}
\begin{proof}
  A subforest of $t$ is a forest $s$ such that $t=ps$ for some context
  $p$.  Take a forest $t \in L_\bot$ and consider a subforest
  $s$ of $t$ that is in $L_\bot$, but has no proper subforests in
  $L_\bot$.  Then either $s$ is a tree $as'$ with $s' \not \in
  L_\bot$, or $s=s_1+s_2$ with $s_1,s_2 \not \in L_\bot$. Therefore, a
  forest is in $L_\bot$  if and only if it has a subforest in
  \[
\bigcup_{\substack{a \in A, g \not \in \botcomp\\ ag  \in \botcomp} }
     a L_g \qquad \cup \qquad \bigcup_{\substack{g_1,g_2 \not \in \botcomp\\ g_1 +g_2  \in \botcomp} }
      L_{g_1} + L_{g_2} \ .
  \]
  Containing such a subforest can be expressed by a $\Sigma_2$ forest
  expression, by prefixing the set above with the set of all contexts.
  The expressions for $L_g, L_{g_1}$ and $L_{g_2}$ are $\Sigma_2$
  forest expressions by the results from the section on types outside $\botcomp$.
\end{proof}

\subsubsection{Proof of Lemma~\ref{lemma:k-approx}}
\label{sec:proof-lemma-k}
Our goal in this section is to prove Lemma~\ref{lemma:k-approx}, which
says that for any context type $u$ stabilizing $v$, there is a
$\Sigma_2$ context expression $\overap K_u$ such that for any forest
$s$ with a type in $\botcomp$, $\overap K_us$ is a
$v$-overapproximation of $K_us$.

We apply Proposition~\ref{prop:context-expression} to get an
expression for the context language $K_u$ of the form
\begin{equation}
 \label{eq:ku-decomp}
 K_u = \bigcup_i K_{i,1} \cdots K_{i,n_i}
\end{equation} The problem with the expression above is that it may use,
in some of the subexpressions $K_{i,j}$, languages
$M_f + \hole$ or $\hole + M_f$  that involve  forest types $f \in \botcomp$, and we do not know
how to describe types in $\botcomp$. This is where the
overapproximation comes in. We show that if the languages for types in
$\botcomp$ are overapproximated, then the result satisfies the
properties required by Lemma~\ref{lemma:k-approx}. A more detailed
argument is described below.

We say a context language $K$ satisfies (*) if it has the property
required from $K_u$ by Lemma~\ref{lemma:k-approx}, namely
\begin{quote}
	(*) there
	is a $\Sigma_2$ context expression $\overap K$ such that for any
	forest $s$ with a type in $\botcomp$, the language $\overap K s$ is a
	$v$-overapproximation of $Ks$.	
\end{quote}
  It is not difficult to see that
property (*) is preserved by unions and compositions of context
languages. In particular, in order to prove
Lemma~\ref{lemma:k-approx}, it suffices to show (*) is satisfied by
all languages $K_{i,j}$ that appear in~(\ref{eq:ku-decomp}).

The only problem with the overapproximation is when $K_{i,j}$ is of
the form $M_f + \hole$ or $\hole + M_f$, for $f \in \botcomp$. In all
the other cases, $K_{i,j}$ is known to be definable by a $\Sigma_2$
context expression, and no overapproximation is needed. Note that
by~(\ref{eq:assumption-left}), the expressions $\hole+L_f$ cannot be
used, since a forest type from $\botcomp$ cannot appear to the right
of the hole in a context type that stabilizes $v$. Therefore, to complete
the proof of Lemma~\ref{lemma:k-approx}, it remains to show that for
any $f \in \botcomp$, the context language $M_f + \hole$ satisfies
(*).

In the following, we use an equivalence relation $\equiv_{v+}$. This is defined
to be the intersection of all equivalence relations $\equiv_{vu}$, for context
types $u$ that do not stabilize $v$ (and hence $v$ is a strict prefix of $vu$).
For a forest language $L$, we write
\begin{equation}
  \label{eq:v+}
    [L]_{v+} = \bigcap_{u \not \in stab(v)} [L]_{vu}
\end{equation}
By the induction assumption in Proposition~\ref{prop:contexts}, each
equivalence class of $\equiv_{v+}$ is definable by a $\Sigma_2$ forest
expression and therefore so is each language $[L]_{v+}$, as a union of
equivalence classes of $\equiv_{v+}$. We will show that, for any $f
\in \botcomp$, the context language $K=M_f + \hole$ satisfies (*) with
$\overap K = [L_f]_{v+}+\hole$.  Assume that the type of $s$ is
 in $\botcomp$. Then we have:
\[
  Ks  = M_f +s \quad \subseteq \quad
  \overap K s = 
  [L_f]_{v+} + s
  \quad \subseteq \quad [Ks]_v = [L_f+s]_v \ .
\]
The first inequality is clear. For the second inequality, we need to
show that 
\[
  v \cdot \alpha([L_f]_{v+} +s) \qquad \subseteq \qquad v \cdot
  \alpha(L_f + s) 
\]
This inclusion holds because we have:
\[
  v \cdot \alpha([L_f]_{v+} +s) = v \cdot (\hole +
  \alpha(s)) \cdot \alpha([L_f]_{v+})  =    v \cdot (\hole +
  \alpha(s)) \cdot \alpha(L_f)  = v \cdot   \alpha(L_f + s) 
\]
In the second equality, we used the definition of $[L_f]_{v+}$ and the
assumption that $\hole + \alpha(s)$ does not stabilize $v$. The latter
follows from assumption~(\ref{eq:assumption-left}) since the type of
$s$ is in $\botcomp$.

\subsubsection{Proof of Lemma~\ref{lemma:y-approx}}
\label{sec:proof-lemma-y}
In this section, we show that for every type $f \in \botcomp$,
a $v$-overapproximation of $Y_f$ can be defined by a $\Sigma_2$ forest
expression. Recall that the language $Y_f$ was defined based on a case
distinction in Lemma~\ref{lemma:decomp}, and therefore it can be decomposed into a union of four languages, one for each of the four cases in the lemma.  As
$v$-overapproximations are closed under union, for each of these
languages, we provide a $v$-overapproximation defined by a $\Sigma_2$
forest expression.

For the languages corresponding to cases (\ref{item:decomp2bis}) and (\ref{item:decomp3bis}), we use the assumption that $L_h$ can be defined by a $\Sigma_2$ expression for every type $h \not \in \botcomp$. The interesting cases are (\ref{item:decomp2}) and~(\ref{item:decomp3}).

The language corresponding to  case (\ref{item:decomp2}) is a union of forest languages of the form
\begin{equation}
 \label{eq:case2}
 a  \cdot  L_h  
\end{equation}
ranging over letters $a$ such that  $a\hole$ does not stabilize
   $v$,  and forest types $h$ with $ah=f$. We treat each language $a \cdot L_h$ separately.

It may be the case that $h$ belongs to $\botcomp$ and therefore we
have no $\Sigma_2$ forest expression for $L_h$. However, we can
use overapproximation.  As $a$ does not stabilize $v$, we can apply
the induction assumption in Proposition~\ref{prop:contexts} and obtain
a $\Sigma_2$ forest expression for $[L_h]_{v+}$.  But then the
$\Sigma_2$ forest expression  $ a \cdot [L_h]_{v+} $ is
a $v$-overapproximation of $a \cdot L_h$. It clearly contain $a \cdot L_h$ so
it remains to show that it is included in $[aL_h]_v$. To see this consider a forest $t \in
[L_h]_{v+}$ of type $g$ and an arbitrary $u \in V$. As $a$ does not stabilize
$v$, $ua$ does not stabilize $v$. Hence by the choice of $g$ we have
$vuag=vuah$ and $at$ is in $[aL_h]_v$.


It remains to consider the case of  (\ref{item:decomp3}),
where have a union of sets
\begin{equation}
 \label{eq:case3a}
 L_{h_1} + L_{h_2} 
\end{equation}
ranging over $h_1,h_2$ that satisfy the two implications in item (\ref{item:decomp3}) of  Lemma~\ref{lemma:decomp}. We do each pair $h_1,h_2$ separately. 
We consider three subcases.

The first subcase is when $h_1 \not\in \botcomp$.
Therefore we have a $\Sigma_2$ forest expression for $L_{h_1}$. In
this case we claim that
\[
  L_{h_1} + [L_{h_2}]_{v+}
\]
is a $v$-overapproximation of of $L_{h_1} + L_{h_2}$.  The first
requirement of  $v$-overapproximation, 
\[
L_{h_1} + L_{h_2} \quad \subseteq \quad  L_{h_1} + [L_{h_2}]_{v+}\ ,
\]
is immediate. For the other requirement, we need to show that for any
forests $t_1 \in L_{h_1}$ and $t_2 \in [L_{h_2}]_{v+}$, the type of
$t_1 + t_2$ is $\equiv_{v}$-equivalent to $h_1 + h_2$.  From $t_1 \in
L_{h_1}$, we know that the type of $t_1$ is $h_1$, but all we know
about $t_2$ is that its type $g_2$ satisfies $g_2 \equiv_{v+}
h_2$. Consider an arbitrary $u\in V$. Since $h_1 + \hole$ does not stabilize
$v$, we also have $u(h_1 + \hole)$ does not stabilize $v$. Hence from $g_2 \equiv_{v+}
h_2$ we get $vu(h_1 + g_2)=vu(h_1+h_2)$.

The second subcase, when $h_2 \not\in \botcomp$, is treated as above by symmetry.

The third subcase is when both $h_1$ and $h_2$ are in $\botcomp$.
As a consequence of $h_2 \in \botcomp$ and the second condition of item~(\ref{item:decomp3}) in Lemma~\ref{lemma:decomp} is that
\begin{equation}
  \label{eq:case-with-intersection}
   \botcomp + \hole \quad \cap \quad  \stab(v) \quad  =  \quad \emptyset\ .
\end{equation}
We claim that a $v$-overapproximation of $L_{h_1}
+ L_{h_2}$ is 
\begin{equation}
  \label{eq:xy-overapprox}
    ([L_{h_1}]_{v+} \cap L_\bot) \quad   + \quad     ([L_{h_2}]_{v+}
    \cap L_\bot)\ .
\end{equation}
As before, the problem boils down to showing that for any forests 
\[
    t_1 \in [L_{h_1}]_{v+} \cap L_\bot  \qquad      t_2 \in [L_{h_2}]_{v+}
    \cap L_\bot\ ,
\]
the types $g_1$ of $t_1$ and $g_2$ of $t_2$ satisfy $g_1 + g_2
\equiv_v h_1 + h_2$. In other words, we need to show that for an arbitrary $u
\in V$ 
\[
  vu(g_1 + g_2) = vu(h_1+h_2)\ .
\]
Since $t_1 \in L_\bot$, then by~(\ref{eq:case-with-intersection}) the context
type $g_1+ \hole$ does not stabilize $v$ and therefore the same holds for
$u(g_1+ \hole)$. Hence, we can use the assumption on $g_2 \equiv_{v+} h_2$
to infer
\[
  vu(g_1 + g_2) = vu(g_1+\hole)g_2 = vu(g_1+\hole)h_2 = vu(g_1 + h_2) \ .
\]
In a similar way, we use $h_2 \in \botcomp$, the
assumption~(\ref{eq:assumption-left}), and $g_1 \equiv_{v+} h_1$, to
complete the proof of this case, and of  Lemma~\ref{lemma:y-approx}.
\[
  vu(g_1 + h_2) = vu(\hole + h_2)g_1 =vu(\hole + h_2)h_1 = vu(h_1+h_2)\ .
\]

\section{No lexicographic order }\label{sec:nolex}
In this section, we consider the logic \Dtwo where only the descendant
order, and not the lexicographic order, is available. We give an
effective characterization in the following theorem.
\begin{thm}[Effective characterization of $\Delta_2$ with  the descendant
order only]\label{thm:main-commutative}\ \\
 A forest language is definable in \Dtwo if and only if its syntactic
  forest algebra satisfies the $\Delta_2$ identity, as well as
  horizontal commutativity:
\begin{equation}
  \label{eq:horziontal-commutativity}
    h+g = g+h \ .
\end{equation}
\end{thm}

The ``only if'' implication is easy: we have already shown that the $\Delta_2$
identity must hold in the syntactic forest algebra of a language definable in
\Dtwol, and \Dtwo is a fragment of \Dtwol. Horizontal commutativity must
also hold: the logic only has the descendant relation, and therefore its
formulas are invariant under rearranging sibling subtrees.

The ``if'' implication is a minor variation on the work done in the
previous sections.  Recall that we proved before that if the syntactic
forest algebra of a language $L$ satisfies the $\Delta_2$ identity,
then both $L$ and its complement can be defined by $\Sigma_2$ forest
expressions. We apply this result also in our case. The problem is
that the $\Sigma_2$ forest expressions are not commutative, and thus
need not be definable in \Stwo.  We will show, however, that their
commutative closure can be defined in \Stwo. Here, we use the term
\emph{commutative closure of $L$} for the smallest language that
contains $L$ and is closed under rearranging sibling subtrees.

\begin{prop}
  The commutative closure of a $\Sigma_2$ forest expression is
  definable in \Stwo.
\end{prop}
Before we prove this proposition, we remark that this is not a
completely generic result. For instance consider the following language over a one letter alphabet: ``Each node is a
leaf or has two children, and some leaf has an even number of
ancestors''. This language is definable in $\Sigma_3(<,\lex)$ and is horizontally
commutative (the formula comes from Potthoff~\cite{potthoff}). However, this language cannot be defined in $\Sigma_3(<)$. Actually,   an Ehrenfeucht-Fraiss\'e argument shows that every   first-order formula, that has  quantifier depth $n$  and only uses the
descendant order, will give the same result for all balanced binary trees of depths larger than $2^n$.

The proof of the above proposition is by induction on the size of the
$\Sigma_2$ forest expression. 

The base case is when the  $\Sigma_2$ expression is either $\set{a
  \hole}$, or a language that is closed under pieces. In the first
case, the language is clearly definable in \Stwo. In the second case,
we revisit the proof of Lemma~\ref{lem:pi1-lemma}, which showed that a
language $L$ that is closed under pieces is definable in $\Pi_1$. If
we take the commutative closure of $L$, we get a commutative language closed
under pieces. In the proof of Lemma~\ref{lem:pi1-lemma}, we
constructed the $\Pi_1$ formula by  forbidding a finite number of
pieces; in the commutative case the formula does not need to worry
about the order of siblings in the forbidden pieces.

In the induction step, we have to consider the operations that are
allowed by $\Sigma_2$ expressions: union, intersection, (horizontal)
concatenation
\begin{equation}
  \label{eq:concat}
    L+L'
\end{equation}
and (vertical) composition
\begin{equation}
  \label{eq:compos}
    K \cdot L \quad K \cdot K'
\end{equation}
for a forest languages $L,L'$ and context languages $K,K'$. Union and intersection are easy. Concatenation and composition are more problematic. Actually, 
$\Stwo$ is not closed under these two operations. For instance, the
languages
\begin{eqnarray*}
  K &= & \set{a+\hole, \hole + a}\\
L &=& \set{b+c,c+b}
\end{eqnarray*}
are both definable in $\Stwo$, but their concatenation
\[
  KL=\set{a+b+c,a+c+b,b+c+a,c+b+a}
\]
is not, since it does not contain the forest $b+a+c$. 

Nevertheless, if we use commutative closure, the problem disappears. That
is, we will show that if the languages $K,K',L,L'$ are definable in
\Stwo, then the commutative closure of each of the languages
in~(\ref{eq:concat}) and~(\ref{eq:compos}) can be defined in \Stwo.
We only do the cases $L+L'$ and $K \cdot K'$, the language $K \cdot L$ is 
done the same way.

\begin{lem}\label{oplus-forests}
  If forest languages $L,L'$ are definable in \Stwo,  then so is the
  commutative closure of $L+L'$.
\end{lem}
\begin{proof}
We write $L \oplus L'$ for the commutative closure of $L + L'$.

Consider first the case when $L'$ is a \emph{tree} language definable in
\Stwo. In this case, the formula for $L \oplus L'$ formula places an
existentially quantified variable over the root of one tree, and then
relativizes the formulas for $L$ and $L'$, respectively, to the nodes
the are not descendants (respectively, are descendants), of this
existentially quantified root.

For the general case, we use the following lemma on forest languages definable in $\Stwo$.
\begin{lem}\label{lem:experimental}
	Every forest language $L$ definable in $\Stwo$ can be written
	as a finite union of languages $L_0 \oplus M_1 \oplus \cdots \oplus M_n$, where
	$L_0$ is a forest language definable in $\Pi_1(<)$, and  $M_1,\ldots,M_n$ are tree
	languages definable in \Stwo.
\end{lem}
\begin{proof}
	The statement of the lemma immediately follows from the following claim on formulas of $\Stwo$. We claim that any formula  $\varphi$  of $\Stwo$ is equivalent to a finite disjunction of formulas of the form
	\[
          \exists z_1 \ldots \exists z_n \quad \psi\, \land \bigwedge_{i \in
            \set{1,\ldots,n}} (\forall y \ \neg (y < z_i)) \land \psi_i
	\]
	where  $\psi \in \Pi_1(<)$,  and $\psi_1,\ldots,\psi_n \in \Stwo$ are formulas such that
	\begin{enumerate}[$\bullet$]
		\item Each formula $\psi_i$ has all  quantification  relativized to descendants of $z_i$.
		\item The formula $\psi$ has all quantification relativized to nodes that are not descendants of any of the nodes $z_1,\ldots,z_n$.
	\end{enumerate}
	The idea, of course, is that the $z_i$ describe the roots of the trees
        that contain the existentially quantified nodes $x_1,\ldots,x_m$ in the
        original formula $\varphi$ of $\Stwo$. The straightforward proof of the
        claim is omitted. The finite disjunction ranges over all possible
        repartitions of the nodes $x_1,\ldots,x_m$ into distinct trees.
\end{proof}

From this normal form, since $\oplus$ distributes across union, it
suffices to give a \Stwo formula for languages of the form
  \[
L_0 \oplus M_1 \oplus \cdots \oplus M_n \quad \oplus \quad      L'_0
\oplus M'_1 \oplus \cdots \oplus M'_n\ ,
  \]
  where the $M$ languages are tree languages definable in \Stwo and where
  $L_0,L'_0$ are forest languages definable in $\Pi_1(<)$. By the technique
  shown at the beginning of the proof, it is sufficient to obtain a formula for
  $L_0 \oplus L'_0$. But the language $L_0 \oplus L'_0$ is closed under pieces,
  and therefore it is definable in $\Pi_1(<)$.
\end{proof}

\begin{lem}
  If context languages $K,K'$ are definable in \Stwo,  then so is the
  commutative closure of $K\cdot K'$.
\end{lem}
\begin{proof} 
We write $K \odot K'$ for the commutative closure of $K \cdot K'$.

Consider first the case when either $K$ or $K'$ is a language $\set{a
  \hole}$. The formula places a variable on node corresponding to $a
\hole$, and relativizes the formula for the remaining context language
to the remaining nodes. 

Consider now the general case. Again, we use a normal form lemma for languages definable in $\Stwo$.  This lemma is prove the same way as Lemma~\ref{lem:experimental}.
\begin{lem}\label{lem:experimental2}
	Every context language $K$ definable in $\Stwo$ can be  written as  a finite union
	of languages of the kinds
	\[
	  \hat K \odot \set{a \hole} \odot   (\hole \oplus L) 
	\qquad \mbox{or}\qquad 
	\hole \oplus L
	\]
	where $\hat K$ and $L$ are context language and forest languages definable in
	\Stwo.	
\end{lem}
%

We now use Lemma~\ref{lem:experimental2} to finish the proof of Lemma~\ref{lem:experimental2}. We want to show that $K \odot K'$ is definable in $\Stwo$. We apply Lemma~\ref{lem:experimental2} to the languages $K$ and $K'$.
Since the operation $\odot$ distributes across union, we can assume
that  the unions describing $K$ and $K'$ use just one language, of either of the two kinds. We have four cases to consider, we only do the most difficult
one
\[
    \hat K \odot \set{a \hole} \odot   (\hole \oplus L)  \quad \odot
    \quad (\hole \oplus L') \odot \set{a' \hole} \odot     \hat K'\ .
\]
This language is the same as 
\[
    \hat K \odot \set{a \hole} \odot   (\hole \oplus L \oplus L')  \odot \set{a' \hole} \odot     \hat K'\ .
\]
Let $\psi_{\hat K}$ be the \Stwo formula defining the context language $\hat K \odot \set{a \hole}$, 
obtained from the first case considered in the proof.  Let $\psi_{\hat K'}$ be the \Stwo formula
defining $\set{a' \hole} \odot \hat K'$ obtained in the same way. Both of these formulas have a free variable, which describes the hole of the context. Using
Lemma~\ref{oplus-forests} we also have a \Stwo formula $\psi_{L \oplus L'}$
defining $L \oplus L'$.

The desired \Stwo formula puts an existentially quantified variable $x$ on the
node corresponding to $a\hole$, another existentially quantified
variable $x'$ on the node corresponding to $a'\hole$, and then runs
three subformulas, for $\hat K$, $\hole \oplus L \oplus L'$, and $\hat
K'$, on the remaining nodes, appropriately relativizing the quantification. More specifically, this is the formula
\[
  \exists x \ \exists x'  \quad a(x) \wedge a'(x) \wedge \text{parent}(x,x') \wedge \varphi_{\hat K}(x) \wedge
  \varphi_{\hat K'}(x') \wedge \varphi_{L \oplus L'}(x,x')
\]
where $\text{parent}(x,x')$ is the $\Pi_1$ formula stating that $x$ is a proper ancestor of $x'$ and there are no nodes in between, $\varphi_{\hat K}(x)$ is constructed from $\psi_{\hat K}$ by relativizing
all quantification to the ancestors of $x$, $\varphi_{\hat K'}(x)$ is constructed from $\psi_{\hat K'}$ by relativizing
all quantification to the descendants of $x$, and $\varphi_{L \oplus L'}(x,x')$ is
constructed from $\psi_{L \oplus L'}$ by relativizing all quantification to the
nodes that are neither ancestor of $x$ nor descendant of $x'$.
\end{proof}

\newcommand{\ign}[1]{}
\ign{
We now want to show Proposition~\ref{prop:bottom-up} in the case where
$H$ is commutative. For this we need to find for each type $h$ a \Stwo
formula for the language $L_h$ of forests of type $h$.  We revisit the
proof of Proposition~\ref{prop:bottom-up} for \Stwol presented in the
previous sections for the non commutative case and adapt it to the
case of \Stwo.  The structure of the proof is the same, it goes by
induction on the size of the algebra $(H,V)$ and then on the position
of $h$ in the reachability pre-order define earlier. Again we have to
distinguish the case where $h$ is minimal from the other cases. As the
algebra still satisfies the identity~(\ref{eq:swallow}) all the
properties that were obtained as a consequence of this identities
remain valid. In particular
Lemmas~\ref{lemma:closed-under-comp},~\ref{lemma:closed-under-piece},~\ref{lem:sub-forest-algebra}
and the results of Section~\ref{sec:treat-cont-like-words} remain
true.  However we are now limited to the composition properties stated
in Lemma~\ref{fact:simple}, and not the more powerful ones stated in
Lemma~\ref{fact:simple-lex}.

\subsection{Non minimal types}
\label{sec:induction-step-nolex}

Let $h$ be a non minimal type. We distinguish again two cases depending on whether
$h+h \sim h$ or not.

\subsubsection{$h+h \sim h$}
\label{sec:h-preserves-itself-nolex}

Let $G$ be the set of pieces of
$h$. Lemmas~\ref{lem:sub-forest-algebra}
and~\ref{lemma:closed-under-piece} remain true, so $(G,\stab(h))$ is a
strict sub-algebra of $(H,V)$ and we obtain by induction a \Stwo
formula for $\alpha^{-1}(h)$.

\subsubsection{$h+h \not\sim h$}
\label{sec:h-does-not-preserve-itself-nolex}

Recall that for $v \in V$, we write $K_v$ for the set of contexts of type $v$
and  that for $g \in H$, we write $L_g$ for the set of forests of type $g$. For $g \in H$
and $F \subseteq H$, we write $L^F_g$ for the set of forests $t$ of type $g$
that can be decomposed as $t=t_1+\ldots+t_n$, with each $t_i$ a tree with of
type in $F$.

Let $G$ be the set of forest types $g$ such that $h$ is reachable from
$g$ but not vice-versa.  By induction assumption, each language $L_g$
with $g \in G$ is definable in \Stwo. Our goal is to give a
formula for $L_h$.

\begin{lem}\label{lem:decompo-nonminim-nolex}
A forest has type $h$ if and only if it belongs to $L^G_h$ or a
language $K_u a L^G_g$, with $u\alpha(a)g = h$ and $u$ stabilizing $h$.
\end{lem}
\vspace{-.2cm}
\begin{proof}
  Let $t$ be a forest of type $h$, and choose $s$ a subtree of $t$ that has a
  type equivalent to $h$, but no subtree with a type equivalent to $h$. If such
  $s$ does not exist, then $t$ belongs to $L^G_h$ as a concatenation of trees
  with type in $G$.  By minimality, $s$ must belong to $a L^G_g$ for some $g\in
  G$. Let $p$ be the context such that $t=ps$.  Since the type of $s$ is
  equivalent to $h$, and the type of $t$ is $h$, then the of $p$ must
  stabilize $\alpha(s)$,  and therefore also stabilize $h$ by
  Lemma~\ref{lemma:closed-under-comp}.
\end{proof}

In Lemmas~\ref{lemma:sigma2-forest} and \ref{lemma:cont-def}, we will
show that the languages $K_u$ and $L^G_g$ above can be defined in
\Stwo. Proposition~\ref{prop:bottom-up} then follows by closure of \Stwo
under finite union and Lemma~\ref{fact:simple}.
To be more precise, we only give an approximation
$\varphi^G_g$ of the language $L^G_g$, however all forests in the
approximation have type $g$, which is all we need.

As in Section~\ref{sec:h-does-not-preserve-itself}, the fact that $K_u$ can be
expressed in \Stwo is a consequence of
Proposition~\ref{prop:context-expression-nolex}.

\begin{lem}\label{lemma:cont-def}
 If   $u \in V$  stabilizes $h$ then the context language $K_u$ is definable in \Stwo.
\end{lem}
\begin{proof}
  From Proposition~\ref{prop:context-expression-nolex} and its proof,
  we know that $K_u$ can be defined by a finite union of expressions
  of the form $K_{1} \cdots K_{m}$ where each $K_{i}$ is either
  definable in \Stwo, or of the form $\hole + M_g \ \cup\ M_g + \hole $ for some type $g
  \in H$, or of the form $(B \cap \beta^{-1}(W))^*$ where $B$ is the
  infinite alphabet of letters of the form $a \hole$ or $t + \hole$,
  for $t$ a tree and $a\ in A$ and $W$ a subset of $V$ closed under
  pieces. In this later case we know that $K_i$ is definable in
  $\Pi_1(<)$.
 
  We are looking for a \Stwo formula for $K_u$. By closure under union of \Stwo
  and item 3 of Lemma~\ref{fact:simple} it is sufficient to obtain a \Stwo
  formula for any expression of the form $K_{1} \cdots K_{m}$ where each $K_i$
  is either $\hole + M_g$ for some type $g \in H$ or of the form $(B \cap
  \beta^{-1}(W))^*$.

  Consider an expression
  \[
    K_i=(B \cap \beta^{-1}(W_i))^*
  \]
 Split  the contexts that form the  letters in $B$
  into two  parts: the set $C \subseteq B$  of  those contexts of the
  form $a\hole$ and the set  $D \subseteq B$ of contexts of the form
 $\hole +t$ or $t+ \hole$. If we set $C_i$ to be $C \cap
 \beta^{-1}(W_i)$ and likewise for $D_i$, we can decompose $K_i$ as 
  \[
     K_i=(C_i \cup D_i)^*=C_i^* \cup
  D_i^*C_iD_i^* \cup D_i^*C_iK_iC_iD_i^* 
  \]
 As context
  languages included in $C$ can be used for composition by item 3 of
  Lemma~\ref{fact:simple}, it remains to show that context languages of the form
  \[
      B_1^* \oplus B_2^* \oplus \cdots B_m^* \oplus M_{g_1} \oplus M_{g_2} \oplus \cdots M_{g_k}
  \]
 are
  definable in \Stwo, where $B_i \subseteq (D)^*$.

  Consider one of the tree language $M_{g_i}$. Since $u\in \stab(h)$, $g_i+\hole \in
  \stab(h)$ and $g_i \sqsubseteq h$. But $g_i$ is not reachable from $h$ because
  otherwise, as $\stab(h)$ is closed under pieces, $\hole +h$ would be in
  $\stab(h)$ and that would contradict $h+h \not\sim h$. Hence by the induction
  assumption from Proposition~\ref{prop:bottom-up}, $M_{g_i}$ is definable in
  \Stwo.

  Now the \Stwo formula for $B_1^* + B_2^* + \cdots B_m^* + M_{g_1} + M_{g_2} +
  \cdots M_{g_k}$ guesses the roots of the trees in $M_{g_1} + M_{g_2} +
  \cdots M_{g_k}$ and verify that any other tree belongs to  $B_1^* + B_2^* +
  \cdots B_m^*$. The former property is in \Stwo by the remark above and the
  latter can be shown to be in $\Pi_1(<)$ as in the proof of Lemma~\ref{lem:pi1}.
\end{proof}

We now construct a \Stwo formula approximating $L^G_g$.

\begin{lem}\label{lemma:sigma2-forest}
  For any type $g \in H$, there is a formula  $\varphi^G_g$ of 
  \Stwo such that:
\begin{itemize}
  \item Any forest  $L^G_g$ satisfies $\varphi^G_g$ and,
\item any forest satisfying $\varphi^G_g$ has type $g$.
\end{itemize}
\end{lem}
\vspace{-.2cm}
\begin{proof}
  The proof of the lemma is in two steps.  In the first step, we
introduce a condition (*) on a forest $t$, and show that: a) any
forest in $L^G_g$ satisfies (*); and b) any forest satisfying (*) has
type $g$. Then we will show that condition (*) can be expressed in
\Stwo. Recall that we are now dealing with an algebra such that $f+g=g+f$.

 \begin{quote}
(*) For some $m \le n$, the forest $t$ can be decomposed, modulo
    commutativity, as the concatenation $t=t_1+\cdots+t_n$ of trees
    $t_1,\ldots,t_n$, with types $g_1,\ldots,g_n$, such that
  \begin{enumerate}
\item $g_1+\cdots+g_m=g$.
\item Each type from $G$ is represented at most $\omega$ times in
$g_1,\ldots,g_m$.
  \item If a tree $s$ is a piece of $t_{m+1}+\cdots + t_n$, then
    $\alpha(s) \piece g_i$ holds for some type $g_i$ that occurs
    $\omega$ times in the sequence $g_1,\ldots,g_m$.
  \end{enumerate}
  \end{quote}

We first show that condition (*) is necessary. Let $t_1,\ldots,t_n$ be
all the trees in a forest $t$, and let $g_1,\ldots,g_n \in G$ be
the types of these trees.  Without loss of generality, we may assume
that trees are ordered so that for some $m$, each type of $g_i$ with
$i > m$ already appears $\omega$ times in $g_1,\ldots,g_m$. It is not hard to see that identity~(\ref{eq:swallow}) implies aperiodicity of the monoid $H$, i.e.~
\begin{equation}\label{eq:aperio}
\omega f = \omega f + f \qquad \mbox{for all }f \in H\ .
\end{equation} 
In particular, it follows that $g=g_1+\cdots +g_m$ since all of
$g_{m+1},\ldots,g_n$ are swallowed by the above. It remains to show
item 3 of condition (*). Let then $s$ be the piece of a tree $t_i$
with $i>m$. We get the desired result since the type of $t_i$ already
appears in $g_1,\ldots,g_m$.

We now show that condition (*) implies $\alpha(t)=g$. Let then $m \le
n$ and $t=t_1+\cdots+t_n$ be as in (*). We will show that for any $j >
m$, we have $g+g_j=g$, which shows that the type of $t$ is $g$. By
item 3, $g_j \piece g_i$ holds for some some type $g_i$ that occurs
$\omega$ times in the sequence $g_1,\ldots,g_m$. By~(\ref{eq:aperio}),
we have $g=g+g_i=g+ \omega g_i$. It therefore remains to show
that $\omega g_i + g_j = \omega g_i$:
\[
\omega g_i + g_j = \omega g_i + g_j + \omega  g_i =\\
(\hole + g_i)^\omega(\hole + g_j) (\hole + g_i)^\omega 0 = (\hole +
g_i)^\omega 0 = \omega \cdot g_i
\]
In the above we have used identity~(\ref{eq:swallow}). Note that the
requirement in~(\ref{eq:swallow}) was satisfied, since $g_j \piece
g_i$ implies $\hole + g_j \piece \hole+g_i$.

It now remains to show that forests satisfying condition (*) can be defined in
\Stwo. Note that $m$ cannot exceed $|G| \cdot \omega$, and therefore there
is a finite number of cases to consider for $g_1,\ldots,g_m$. Fix some sequence
$g_1,\ldots,g_m$.  The only nontrivial part is to provide a \Stwo formula
that describes the set of forests $t_{m+1} + \ldots + t_n$ that satisfy item 3
of condition (*). From this construction, the formula for (*) follows by
closure of \Stwo under finite union and $\oplus$ (recall
Fact~\ref{fact:simple}), as well as the assumption that each type in $G$ can be
defined in \Stwo.

In order to define forests as in item 3 we use a $\Pi_1(<)$ formula that forbids
the appearance of certain pieces of bounded size inside $t_{m+1} + \cdots +
t_n$. Let $F$ be the types in $g_1,\ldots,g_m$ that appear at least $\omega$
times. We claim that a sequence of trees $t_{m+1}+\cdots +t_n$ satisfies item 3
if and only if it satisfies item 3 with respect to pieces $s$ that have at most
$|H|^{|H|}$ nodes. The latter property can be expressed by a $\Pi_1(<)$ formula.
The reason for this is that, thanks to a pumping argument, any tree has a piece
that has the same type, but at most $|H|^{|H|}$ nodes.\luc{do we need more
  details for the pumping argument?}
\end{proof}

\subsection{Minimal types}

It remains to compute a \Stwo formula for $L_h$ where $h$ is a minimal type.
As in Section~\ref{sec:top-down} this is done using another induction on
contexts, using the same pre-order as in Section~\ref{sec:top-down}, and taking
$v=\hole$ in the proposition below. In the remaining part of this section we
denote by $F$ the set of minimal types of $H$ and by $G$ the set of non minimal
types of $H$.

\begin{prop}\label{prop:contexts-nolex}
  For any context \hhl{type} $v$, every equivalence class of $\equiv_v$ is definable in
  \Stwo.
\end{prop}

The rest of this section is devoted to prove this proposition. We first prove
the base case and then move to the induction step. Note that the structure of
the proof is the same as for Section~\ref{sec:top-down} but this proof has to
be modified in order to use weaker compositions properties.

\subsubsection{ The induction base}

The base case is when $v$ is minimal. In this case the context \hhl{type} $v$ is constant
and hence $\equiv_v$ is easily definable in \Stwo.  The fact that $v$ minimal
implies $v$ constant is a simple consequence of the following lemma, itself a
consequence of commutativity of $H$ and
Lemma~\ref{lemma:at-least-one-disjoint}.

\begin{lem}\label{lemma:at-least-one-disjoint-nolex}
  If $F+ \hole$ intersect $\stab(v)$, then the context \hhl{type} $v$ is constant.  
\end{lem}

\subsubsection{The induction step}

We now proceed to the induction step, and assume that $v$ is not
minimal. As previously let $F$ be the set of minimal types. Recall that $L_g$
denotes the forests of type $g$. We write $M_g$ the trees of type $g$.

As in Section~\ref{sec:top-down}, Proposition~\ref{prop:contexts-nolex} follows
from the following Proposition.

\begin{prop}\label{prop:overapprox-nolex}
 For every $h \in H$, some $v$-overapproximation of $L_h$ is definable in \Stwo.
\end{prop}

As before, this proposition is simple for $h$ not minimal as no
overapproximation is needed in this case. We therefore only consider the case
where $h$ is minimal. In this case as we can only use Lemma~\ref{fact:simple}
for composing languages we need to decompose minimal forests in a more precise
way then what we did in Lemma~\ref{lemma:decomp}.

\begin{lem}\label{lemma:decomp-context-nolex}
Any forest of minimal type satisfies one of the two cases below:

\noindent
In the first case, the forest is of the form $uas$ where $u \in
\stab(v)$ and either
\begin{enumerate}
 \item [(i)] $s'=0$
 \item [(ii)] the type of $a$ is not in $\stab(v)$,
 \item [(iii)] $s=s_1+s_2$ where $s_1$ is a \emph{tree} of type $g+\hole \not\in \stab(v)$ and $g$ not minimal or,
 \item [(iv)] $s=s_1+s_2$ where $s_1$ is a tree of minimal type and $s_2$ is a forest
   of minimal type.
\end{enumerate}

\noindent
In the second case, it is a forest $s$ verifying item (iii) and (iv) above.
\end{lem}
\begin{proof}
  Let $t$ be a forest of minimal type. From Lemma~\ref{lemma:decomp} we know
  that $t$ is of the form $us$ with $u\in \stab(v)$ and $s$ in one of the forms
  (1)-(4) mentioned in this lemma. If (1) or (2) hold, $s=as'$ with $s'=0$ or
  $a \not \in stab(v)$ and correspond to items (i) and (ii) above. It (3) holds
  $s=s_1+s_2$ with the type of $s_1$ being $g+\hole \not\in\stab(v)$ and $g$ not
  minimal. Therefore the first tree of $s_1$ has a type satisfying condition
  (iii) of the lemma. If this tree has no parent in $u$ we are in the second
  case item (iii) of the lemma, otherwise we are in the first case item (iii).
  Otherwise (4) hold and $s=s_1+s_2$ where both have minimal types. Therefore
  (iv) holds and depending on whether $s$ has a parent in $u$ we are in the
  first or second case of the lemma.
\end{proof}

From Lemma~\ref{lemma:decomp-context-nolex} we have:

\[
L_h=\bigcup_{\substack{u \in \stab(v)\\a\not\in\stab(v)\\uaf=h}} K_uaL_f 
\bigcup_{\substack{u \in \stab(v)\\a\in\stab(v)\\uaf=h}} K_uaY_f \cup Y_h
\]

where $Y_f$ is the set of forests of type $f$ verifying conditions (iii) and
(iv) of Lemma~\ref{lemma:decomp-context-nolex}.

\noindent
Because $Y_f$ is of the form $M_g \oplus L_f$ for some $g$ and $f$, it follows
immediately from the proof of Lemma~\ref{lemma:y-approx} that there exists a
$v$-overapproximation $Y^\uparrow_f$ of $Y_f$ that is definable in \Stwo. By
induction we also have a $vua$-overapproximation $L^{ua}_f$ of $L_f$ for any
$a$ such that $a\not\in \stab(v)$. From Lemma~\ref{lemma:proper-stabilizer}
below we know that $K_u$ is definable in \Stwo. By putting all this together
this immediately yield a $v$-overapproximation of $L_h$ that is definable in
\Stwo by Lemma~\ref{fact:simple}.

\begin{lem}  \label{lemma:proper-stabilizer}
If $u \in \stab(v)$, then $K_u$ is definable in $\Sigma_2$.
\end{lem}
\begin{proof}
  The proof is the same as for Lemma~\ref{lemma:cont-def}.
  Recall that from Proposition~\ref{prop:context-expression-nolex} we can decompose
  $K_u$ as a finite union of expressions of the form $K_{1}
  \cdots K_{m}$ where each $K_{i}$ is either of the form $a\hole$, or $\hole +
  M_g$ for some type $g \in H$, or is definable in $\Pi_1(<)$.

  We show that because $u$ is not minimal then the blocks $\hole + M_g$ must be
  such that $g$ is not minimal, and hence $M_g$ is definable in \Stwo by
  induction on Proposition~\ref{prop:bottom-up}. Once this is done the rest of
  the proof is copied from the proof of Lemma~\ref{lemma:cont-def}.

  Assume to the contrary that $u=u_1(\hole + f)u_2$ with $f$ minimal. As $u \in
  \stab(v)$ then $\hole +f$ is in $\stab(v)$. From
  Lemma~\ref{lemma:at-least-one-disjoint-nolex} this implies that $v$ is
  minimal.
\end{proof}
}

\section{Discussion}
\label{sec:conclusion}
In this paper we considered a signature with the descendant and
lexicographic orders. It would be interesting to know what happens in
the presence of other predicates such as the closest common ancestor,
next sibling or child.

Probably the most natural continuation of this work would be an
effective characterization of \Stwo or \Stwol. Note that this would
strengthen our result: a language $L$ is definable in $\Delta_2$ if
and only if both $L$ and its complement are definable
in~$\Sigma_2$. We conjecture that, as in the case for
words~\cite{arfi91}, the characterization of \Stwol requires replacing
the equivalence in the $\Delta_2$ identity by a one-sided implication,
which says that a language definable in \Stwol is closed under
replacing $v^\omega$ by $v^\omega w v^\omega$, for $w \piece v$.


\begin{thebibliography}{10}
\bibitem{arfi91}
M.~Arfi.
\newblock Op{\'e}rations polynomiales et hi{\'e}rarchies de concat{\'e}nation.
\newblock {\em Theor. Comput. Sci.}, 91(1):71--84, 1991.

\bibitem{segoufinfo}
M.~Benedikt and L.~Segoufin.
\newblock {Regular tree languages definable in FO and in FO+mod}.
\newblock To appear in {\em ACM Transactions on Computational Logic
  (TOCL)}. 2009.



\bibitem{fo2tree}
M.~Boja\'nczyk.
\newblock Two-way unary temporal logic over trees.
\newblock In {\em Logic in Computer Science}, pages 121--130, 2007.

\bibitem{forestexp}
M.~Boja\'nczyk.
\newblock Forest expressions.
\newblock In {\em Computer Science Logic}, volume 4646 of {\em Lecture Notes in
  Computer Science}, pages 146--160, 2007.

\bibitem{efextcs}
M.~Boja\'nczyk and I.~Walukiewicz.
\newblock Characterizing {EF} and {EX} tree logics.
\newblock {\em Theoretical Computer Science}, 358(2-3):255--273, 2006.

\bibitem{forestalgebra}
M.~Boja\'nczyk and I.~Walukiewicz.
\newblock Forest algebras.
\newblock In {\em Automata and Logic: History and Perspectives}, pages 107 --
  132. Amsterdam University Press, 2007.

\bibitem{simontrees}
M.~Boja\'nczyk,  L.~Segoufin, and H.~Straubing.
\newblock Piecewise testable tree languages.
\newblock In {\em Logic in Computer Science}, 2008

\bibitem{EVW02}
K.~Etessami, M.~Y. Vardi, and T.~Wilke.
\newblock First-order logic with two variables and unary temporal logic.
\newblock {\em Inf. Comput.}, 179(2):279--295, 2002.


\bibitem{mcnaughton}
R.~McNaughton and S.~Papert.
\newblock {\em Counter-Free Automata}.
\newblock MIT Press, 1971.

\bibitem{pin-survey}
J.-\'E. Pin.
\newblock Logic, semigroups and automata on words.
\newblock {\em Annals of Mathematics and Artificial Intelligence}, 16:343--384,
  1996.

\bibitem{weilpinpoly}
J.-\'E. Pin and P.~Weil.
\newblock Polynomial closure and unambiguous product.
\newblock {\em Theory Comput.~Systems}, 30:1--30, 1997.

\bibitem{potthoff}
A.~Potthoff
\newblock First-order logic on finite trees.
\newblock In {\em TAPSOFT}, volume 915 of {\em Lecture Notes in
  Computer Science}, pages 125--139, 1995.


\bibitem{schutzenberger}
M.~P. Sch\"utzenberger.
\newblock On finite monoids having only trivial subgroups.
\newblock {\em Information and Control}, 8:190--194, 1965.

\bibitem{turtle}
T.~Schwentick, D.~Th\'erien, and H.~Vollmer.
\newblock Partially-ordered two-way automata: A new characterization of DA.
\newblock In {\em Devel. in Language Theory}, pages 239--250, 2001.

\bibitem{simonpiecewise}
I.~Simon.
\newblock Piecewise testable events.
\newblock In {\em Automata Theory and Formal Languages}, pages 214--222, 1975.

\bibitem{therienwilkefo2}
D.~Th\'erien and T.~Wilke.
\newblock Over words, two variables are as powerful as one quantifier
  alternation.
\newblock In {\em STOC}, pages 256--263,
  1998.

\bibitem{wilke}
T.~Wilke.
\newblock Classifying discrete temporal properties.
\newblock In {\em Symposium on Theoretical Aspects of Computer Science}, volume
  1563 of {\em Lecture Notes in Computer Science}, pages 32--46, 1999.

\end{thebibliography}
\end{document}